\documentclass[aps,prd,reprint,nofootinbib,superscriptaddress,notitlepage,longbibliography]{revtex4-2}

\usepackage[T1]{fontenc}
\usepackage[utf8]{inputenc}
\usepackage[english]{babel}
\usepackage{amsmath,amssymb,mathtools,bm}
\usepackage{mathrsfs}
\usepackage{amsthm}
\usepackage{booktabs}
\usepackage{microtype}
\usepackage{xcolor}
\usepackage{hyperref}
\usepackage{orcidlink}
\usepackage{enumitem}
\usepackage{placeins}

\definecolor{navy}{RGB}{0,0,128}

\hypersetup{
  pdftitle={Inverse Reconstruction of Causal Nonlinear Electrodynamics: Functional Families, Spectral Constraints, and Single-Horizon Black Holes},
  pdfauthor={Ariel Guzman, Mohsen Fathi, J. R. Villanueva},
  pdfsubject={
    Causal nonlinear electrodynamics, inverse reconstruction,
    spectral criteria, Einstein--NLED black holes, optical metrics,
    and symmetric hyperbolicity
  },
  pdfkeywords={
    nonlinear electrodynamics, causality, inverse problem,
    black holes, horizon uniqueness, spectral representation,
    optical metrics, symmetric hyperbolicity
  },
  colorlinks=true,
  linkcolor=navy,
  citecolor=navy,
  urlcolor=navy,
  filecolor=navy
}

\newcommand{\dd}{\mathrm{d}}
\newcommand{\e}{\mathrm{e}}

\newcommand{\cCX}{\mathcal{C}_{X}}
\newcommand{\cCrho}{\mathcal{C}_{\rho}}
\newcommand{\cCcausal}{\mathcal{C}_{\mathrm{causal}}}
\newcommand{\cCblind}{\mathcal{C}_{\mathrm{blind}}}
\newcommand{\cCSD}{\mathcal{C}_{\mathrm{SD}}}
\newcommand{\Eem}{E_{\mathrm{em}}}
\newcommand{\Ered}{E_{\mathrm{red}}}
\newcommand{\avg}[1]{\left\langle #1\right\rangle}
\newcommand{\supp}{\operatorname{supp}}
\newcommand{\Tr}{\operatorname{tr}}
\newcommand{\R}{\mathbb{R}}
\newcommand{\Fcal}{\mathcal{F}}
\newcommand{\Lmag}{\mathscr{L}}

\newtheorem{theorem}{Theorem}[section]
\newtheorem{proposition}[theorem]{Proposition}
\newtheorem{corollary}[theorem]{Corollary}

\theoremstyle{remark}
\newtheorem{remark}[theorem]{Remark}

\begin{document}

\title{Inverse Reconstruction of Causal Nonlinear Electrodynamics:\\ Functional Families, Spectral Constraints, and Single-Horizon Black Holes}

\author{Ariel Guzm\'an\orcidlink{0009-0008-3844-1203}}
\email{ariel.guzman@estudiantes.uv.cl}
\affiliation{Institute of Physics and Astronomy, Faculty of Sciences, Universidad de Valpara\'iso, Avenida Gran Breta\~na 1111, Valpara\'iso, Chile}

\author{Mohsen Fathi\orcidlink{0000-0002-1602-0722}}
\email{mohsen.fathi@ucentral.cl}
\affiliation{Center for Research in Space Sciences and Theoretical Physics (CICEF), Universidad Central de Chile, La Serena 1710164, Chile}

\author{J. R. Villanueva\orcidlink{0000-0002-6726-492X}}
\email{jose.villanueva@uv.cl}
\affiliation{Institute of Physics and Astronomy, Faculty of Sciences, Universidad de Valpara\'iso, Avenida Gran Breta\~na 1111, Valpara\'iso, Chile}

\date{\today}

\begin{abstract}
Nonlinear electrodynamics (NLED) admits many causal theories, so causality alone does not provide a unique selection principle. We formulate an inverse construction in which constitutive integrability, the Maxwell weak-field limit, and causal propagation are imposed before either a Lagrangian or a spacetime geometry is chosen. An affine-separable reduction of the two-invariant Pleba\'nski class yields an infinite-dimensional causal family $\cCX$, characterized by a bounded logarithmic index; Born--Infeld is its unique self-dual member, while generic members are birefringent. On the magnetic axis, complete monotonicity gives a positive spectral representation. Finite monopole self-energy is equivalent to the existence of the spectral moment of order $-1/4$, whereas global magnetic causality restricts the support. Generalized-gamma spectra are simultaneously causal and finite-energy precisely for $1/4<\gamma\le1/2$, independently of the shape parameter, and a separate criterion determines when the magnetic law admits a causal two-invariant completion. After coupling to Einstein gravity, a positive magnetic response and characteristic factor, finite self-energy, and nonnegative residual mass imply a strictly increasing metric function, excluding more than one positive horizon; positive residual mass guarantees a unique horizon. The two optical metrics of $\cCX$ remain Lorentzian with overlapping timelike cones, establishing symmetric hyperbolicity of the electromagnetic subsystem. Thus inverse matter selection connects local causal consistency to global black-hole structure without prescribing the geometry.
\end{abstract}

\maketitle

\section{Introduction}
\label{sec:intro}

Nonlinear electrodynamics (NLED) is a useful setting in which consistency requirements on a matter theory can be followed from local field propagation to gravitational backreaction. Born and Born--Infeld originally introduced nonlinear electromagnetic models in connection with the classical self-energy problem \cite{Born1933,BornInfeld1934}; Euler--Heisenberg and Schwinger later established nonlinear electromagnetic response as an effective description of quantum vacuum polarization \cite{HeisenbergEuler1936,Schwinger1951}, and Born--Infeld theory subsequently acquired a string-theoretic interpretation \cite{Tseytlin1999}. Modern analyses add causality, analyticity, positivity, and energy conditions to the criteria by which nonlinear electromagnetic theories are judged \cite{Adams2006,ConzinuUeda2026,BialynickiBirula1983,Sorokin2022}.

For Lorentz-invariant theories of Pleba\'nski type, these requirements are tightly coupled. The Lagrangian fixes both the energy--momentum tensor and the characteristic polynomial governing electromagnetic disturbances. The Fresnel quartic generically factorizes into two quadratic forms and therefore into two optical metrics \cite{Plebanski1970,Boillat1970,Novello2000,DeLorenci2000,ObukhovRubilar2002,HehlObukhovRubilar2002,GoulartBergliaffa2009,deMelo2015}; a regular Maxwell limit does not by itself prevent superluminal propagation, characteristic degeneracy, or loss of hyperbolicity at finite field strength \cite{ShabadUsov2011,Perlick2011,Abalos2015,Schellstede2016}. Necessary and sufficient causality conditions for the Pleba\'nski class were obtained in Ref.~\cite{Schellstede2016}, and the relation between the two effective cones and symmetric hyperbolicity was established in Ref.~\cite{Abalos2015}. Causality also implies the dominant and strong energy conditions under the assumptions of Ref.~\cite{RussoTownsendEnergy2024}. Continuous electric--magnetic self-duality supplies a powerful additional structure \cite{GaillardZumino1981,GibbonsRasheed1995,HatsudaKamimuraSekiya1999,KuzenkoTheisen2001,AschieriFerraraZumino2008}, with recent work developing broad causal self-dual families and their deformations \cite{Bandos2020,RussoTownsendCausalSelfDual2024,RussoTownsendDualities2024,RussoTownsendSimplified2025,Murcia2025,KuzenkoRuhl2026,Chen2025,BabaeiAghbolagh2026,BabaeiAghbolaghVelniHeIsapour2026,BabaeiAghbolaghChenHeHou2026}. By contrast, absence of birefringence is substantially more restrictive and singles out Born--Infeld under the standard regularity and weak-field assumptions \cite{RussoTownsendNoBirefringence2023,MezincescuRussoTownsend2024,GibbonsHerdeiro2001}.

The gravitational sector introduces a separate consistency problem. Einstein--NLED systems admit many charged and magnetic black-hole geometries \cite{Hoffmann1935,PellicerTorrence1969,GarciaSalazarPlebanski1984,deOliveira1994,Breton2003,FernandoKrug2003,Dey2004,CaiPangWang2004,AyonBeatoGarcia1998,AyonBeatoGarcia2000,Bronnikov2001,FanWang2016,Kruglov2017,Yang2022,Bronnikov2022}, yet regularity of a metric does not guarantee admissibility of its matter source \cite{BokulicJuricSmolic2022,BokulicJuricSmolicConundrum2026,DeFeliceTsujikawa2025}. Recent results make the causal restrictions particularly concrete: they constrain black-hole thermodynamics \cite{AbeMedevielleNoumiYoshimura2026}, can remove Cauchy horizons when the electromagnetic self-energy is finite \cite{HaleHennigarKubiznak2026}, and restrict the possible horizon structures of spherically symmetric Einstein--NLED solutions \cite{RussoTownsendBH2026}. These developments motivate a construction in which the matter sector is filtered before the Einstein equations are solved.

That ordering is the defining point of the present work. Inverse NLED constructions usually infer a matter model from a prescribed geometry, electromagnetic profile, or gravitational mapping \cite{BokulicReverse2024,MkrtchyanSvazas2022,AfonsoOlmoOraziRubiera2018,OvgunPantigSaavedra2026}; reconstruction of an action from integrable constitutive relations is also known in a more general electromagnetic setting \cite{AschieriFerrara2013}. Here the primary variables are instead the first-order constitutive responses. Exactness, Maxwell normalization, and the full Pleba\'nski causality inequalities are imposed before any named NLED model or target spacetime is selected. The first question is therefore whether these requirements determine a unique theory; the second is whether analytically tractable sectors of the resulting theory space lead to general gravitational consequences.

The answer to the first question is negative: strictly causal theories possess functional neighborhoods of equally admissible deformations. Closed-form families therefore require an additional reduction rather than following uniquely from the basic axioms. We use a normalized affine-separable reduction because its integrability condition becomes an exactly solvable transport problem. This produces the non-self-dual family $\cCX$ and permits the causal region to be characterized without choosing a Lagrangian ansatz in advance. Independently, the magnetic projection admits a spectral formulation under complete monotonicity, which makes self-energy, support properties, and causal completion accessible through precise integral criteria. Gravity is introduced only after these matter-sector constructions, and the optical analysis is performed last as an independent check of the characteristic problem.

The relation to our previous work is important for the novelty of the present manuscript. Ref.~\cite{FathiGuzmanVillanueva2026} treated the power-kernel branch $\eta=1$ and established the associated magnetic self-energy window and monotonicity mechanism. The present analysis is not a reparametrization of that branch: it adds a full two-invariant functional class, general positive spectral measures, an exact magnetic-to-two-invariant lifting criterion, generalized-gamma spectra for arbitrary $\eta>0$, a monotonicity theorem that is independent of the affine-separable parametrization, and a two-cone hyperbolicity analysis. Table~\ref{tab:previous-present} records this separation explicitly.

\begin{table*}[t]
\caption{Results inherited from Ref.~\cite{FathiGuzmanVillanueva2026} and results established in the present work.}
\label{tab:previous-present}
\centering
\small
\begin{ruledtabular}
\begin{tabular}{p{0.16\textwidth}p{0.35\textwidth}p{0.42\textwidth}}
\textbf{Aspect}
&
\textbf{Ref.~\cite{FathiGuzmanVillanueva2026}}
&
\textbf{Present work}
\\
\hline
Electromagnetic sector
&
Positive mixtures of power kernels
&
Infinite-dimensional non-self-dual class $\cCX$ and general positive spectral measures
\\
Causality and self-energy
&
Window $1/4<\gamma\le1/2$ on the power branch
&
Exact $-1/4$ moment criterion, support restrictions, and two-invariant lifting criterion
\\
Integrable family
&
Branch $\eta=1$
&
Generalized-gamma hierarchy for arbitrary $\eta>0$, including Mellin--Barnes, Meijer--$G$, and Fox--$H$ representations
\\
Gravity
&
Metric monotonicity for the magnetic family considered there
&
Sufficient monotonicity theorem for general magnetic responses satisfying the stated causal and energetic hypotheses
\\
Optics
&
Optical branches of the particular magnetic model
&
Two optical metrics and symmetric hyperbolicity for the full class $\cCX$
\\
\end{tabular}
\end{ruledtabular}
\end{table*}

Section~\ref{sec:causal} fixes conventions and the causal domain. Sections~\ref{sec:blind} and \ref{sec:UV} formulate the inverse problem, while Sec.~\ref{sec:selfdual} records the known self-dual hypersurface used as a control. Section~\ref{sec:CX} constructs the affine-separable class, Sec.~\ref{sec:spectral} develops the spectral sector and its lifting problem, and Secs.~\ref{sec:einstein}--\ref{sec:exactBH} introduce gravity. Section~\ref{sec:optical} analyzes the characteristic cones. Discussion, scope, and conclusions are separated from the technical derivations, which are collected in the appendices.

\section{Two-invariant NLED and causality}
\label{sec:causal}

\subsection{Conventions and physical assumptions}

We work in four dimensions with signature $(-,+,+,+)$ and units $c=\hbar=1$. The action is
\begin{equation}
I[g,A]=\frac{1}{16\pi G}\int \dd^4x\sqrt{-g}\,R
+\int \dd^4x\sqrt{-g}\,\mathcal L(S,P),
\label{eq:action}
\end{equation}
where
\begin{equation}
S=-\frac14F_{\mu\nu}F^{\mu\nu},
\qquad
P=-\frac18\varepsilon^{\mu\nu\rho\sigma}F_{\mu\nu}F_{\rho\sigma}.
\label{eq:SP}
\end{equation}
In a local inertial frame,
\begin{equation}
S=\frac12(\mathbf E^2-\mathbf B^2),
\qquad
P=\mathbf E\cdot\mathbf B.
\end{equation}
We follow the standard conventions for the Pleba\'nski class \cite{Plebanski1970,Sorokin2022}. We impose parity,
\begin{equation}
\mathcal L(S,P)=\mathcal L(S,-P),
\label{eq:parity}
\end{equation}
vanishing vacuum energy,
\begin{equation}
\mathcal L(0,0)=0,
\label{eq:vacuum}
\end{equation}
and a regular weak-field Maxwell limit,
\begin{equation}
\mathcal L(S,P)=S+O(S^2+P^2),
\qquad (S,P)\to(0,0).
\label{eq:Maxwellweak}
\end{equation}
The first condition restricts the pseudoscalar sector to a parity-invariant theory; the second prevents an implicit vacuum-energy contribution from being absorbed into the electromagnetic sector; and the third ensures that the nonlinear theory connects continuously to Maxwell electrodynamics at weak fields. These are physical selection assumptions for the class studied here, not identities obeyed by every NLED theory. The energy--momentum tensor may be written as \cite{Plebanski1970,RussoTownsendEnergy2024,BokulicThermo2021}
\begin{equation}
T_{\mu\nu}=\mathcal L_S T^{\rm Max}_{\mu\nu}
-\bigl(S\mathcal L_S+P\mathcal L_P-\mathcal L\bigr)g_{\mu\nu}.
\label{eq:Tmunu}
\end{equation}

It is useful to introduce the nonnegative variables
\begin{equation}
U=\frac12\left(\sqrt{S^2+P^2}-S\right),
\qquad
V=\frac12\left(\sqrt{S^2+P^2}+S\right),
\label{eq:UVdef}
\end{equation}
for which
\begin{equation}
U,V\ge0,
\qquad
S=V-U,
\qquad
P^2=4UV.
\label{eq:UVrelations}
\end{equation}
These variables are natural in the modern analysis of causality and duality \cite{RussoTownsendEnergy2024,RussoTownsendCausalSelfDual2024}.

\subsection{General causality conditions}

Under the regularity and weak-field assumptions used in the causal classification of Pleba\'nski theories, the necessary and sufficient conditions for both optical modes to remain causal with respect to $g_{\mu\nu}$ can be written, in our conventions, as \cite{Schellstede2016,RussoTownsendEnergy2024}
\begin{equation}
\mathcal L_S>0,
\label{eq:causal1}
\end{equation}
\begin{equation}
\mathcal L_{SS}\ge0,
\qquad
\mathcal L_{PP}\ge0,
\qquad
\mathcal L_{SS}\mathcal L_{PP}-\mathcal L_{SP}^2\ge0,
\label{eq:causal2}
\end{equation}
and
\begin{equation}
\mathcal L_S>
2U\mathcal L_{SS}+2V\mathcal L_{PP}-2P\mathcal L_{SP}.
\label{eq:causal3}
\end{equation}
In $(U,V)$ variables, the last inequality takes the particularly compact form \cite{RussoTownsendEnergy2024,RussoTownsendCausalSelfDual2024}
\begin{equation}
\mathcal L_U+2U\mathcal L_{UU}<0.
\label{eq:strongUV}
\end{equation}
We denote by $\cCcausal(D)$ the subclass of Lagrangians $\mathcal L\in C^2(D)$ that, in addition to parity \eqref{eq:parity}, vacuum normalization \eqref{eq:vacuum}, and the weak-field limit \eqref{eq:Maxwellweak}, satisfy Eqs.~\eqref{eq:causal1}--\eqref{eq:causal3} on a common physical domain $D$.

Full causality should not be confused with separately imposing energy conditions: under the assumptions of Ref.~\cite{RussoTownsendEnergy2024}, the dominant and strong energy conditions follow from causality.

\begin{proposition}[Convexity of the causal space]
Let $\mathcal L_1,\mathcal L_2\in\cCcausal(D)$ and $0\le\lambda\le1$. Then
\begin{equation}
\mathcal L_\lambda=\lambda\mathcal L_1+(1-\lambda)\mathcal L_2
\end{equation}
belongs to $\cCcausal(D)$.
\end{proposition}
\begin{proof}
Parity, the condition $\mathcal L(0,0)=0$, and the linear Maxwell term are preserved because the combination is linear and the coefficients sum to one. In particular,
\begin{equation}
\mathcal L_\lambda=S+O(S^2+P^2)
\qquad (S,P)\to(0,0).
\end{equation}
Moreover,
\begin{equation}
(\mathcal L_\lambda)_S
=\lambda(\mathcal L_1)_S+(1-\lambda)(\mathcal L_2)_S>0.
\end{equation}
The Hessian of $\mathcal L_\lambda$ is a convex combination of two positive-semidefinite matrices and therefore remains positive semidefinite. Finally,
\begin{equation}
\mathscr C[\mathcal L]
=\mathcal L_S-2U\mathcal L_{SS}-2V\mathcal L_{PP}+2P\mathcal L_{SP}
\label{eq:Cfunctional}
\end{equation}
is linear in $\mathcal L$, so its strict positivity is also preserved. Hence $\mathcal L_\lambda\in\cCcausal(D)$.
\end{proof}

Under domination assumptions that justify differentiation under the integral sign, the same observation gives the continuous version
\begin{equation}
\begin{aligned}
\mathcal L_\mu(S,P)
&=\int_{\mathcal A}\mathcal L_\alpha(S,P)\,\mu(\dd\alpha),\\
\mu&\ge0,
\qquad
\mu(\mathcal A)=1.
\end{aligned}
\label{eq:convexIntegral}
\end{equation}
provided all components belong to $\cCcausal(D)$ on the same domain and there is sufficient domination to interchange the required derivatives with the integral. Convex combinations have also appeared in recent interpolations of nonlinear theories \cite{LiuYang2025}; here the property is used at the level of the full two-invariant class.

\section{Constitutive reconstruction}
\label{sec:blind}

The causality conditions of the previous section are still formulated in terms of a Lagrangian $\mathcal L(S,P)$ whose functional form is unknown. To invert the problem without introducing in advance a preferred combination of invariants, a named theory, or a target geometry, we take the first-order constitutive responses as the primary data,
\begin{equation}
\mathcal A(S,P):=\mathcal L_S,
\qquad
\mathcal B(S,P):=\mathcal L_P.
\label{eq:ABdef}
\end{equation}
The Maxwell weak-field limit fixes
\begin{equation}
\mathcal A(0,0)=1,
\qquad
\mathcal B(0,0)=0.
\label{eq:ABvacuum}
\end{equation}
Parity of the Lagrangian, $\mathcal L(S,P)=\mathcal L(S,-P)$, translates completely into
\begin{equation}
\mathcal A(S,-P)=\mathcal A(S,P),
\qquad
\mathcal B(S,-P)=-\mathcal B(S,P).
\label{eq:ABparity}
\end{equation}
In particular,
\begin{equation}
\mathcal B(S,0)=0.
\label{eq:Baxis}
\end{equation}
No further functional dependence between $\mathcal A$ and $\mathcal B$ is assumed.

The first restriction is not causal but geometric: the two responses must arise from a single Lagrangian potential. If $\mathcal L\in C^2$, equality of mixed derivatives requires
\begin{equation}
\partial_P\mathcal A=\partial_S\mathcal B.
\label{eq:integrabilityAB}
\end{equation}
This is the integrability condition for the inverse problem. Reconstruction of an action from constitutive relations subject to an integrability condition is a known mechanism in a more general electromagnetic setting \cite{AschieriFerrara2013}. What is specific here is not exactness of the 1-form by itself, but its combination with the full Pleba\'nski causality inequalities and with the physical restrictions imposed below.

\begin{theorem}[Constitutive reconstruction without a Lagrangian ansatz]
\label{thm:blind-reconstruction}
Let $\Omega\subset\R^2$ be an open simply connected domain, and let $\mathcal A,\mathcal B\in C^1(\Omega)$ satisfy Eq.~\eqref{eq:integrabilityAB}. Then there exists a function $\mathcal L\in C^2(\Omega)$, unique up to an additive constant, such that
\begin{equation}
\mathcal L_S=\mathcal A,
\qquad
\mathcal L_P=\mathcal B.
\label{eq:ABgradient}
\end{equation}
If in addition $(0,0)\in\Omega$ and the normalization $\mathcal L(0,0)=0$ is imposed, the reconstruction is unique and can be written as
\begin{equation}
\mathcal L(S,P)=
\int_{\Gamma:(0,0)\to(S,P)}
\bigl(\mathcal A\,\dd S+\mathcal B\,\dd P\bigr),
\label{eq:pathL}
\end{equation}
where the value of the integral is independent of the path $\Gamma\subset\Omega$.
\end{theorem}

\begin{proof}
Consider the 1-form
\begin{equation}
\omega=\mathcal A\,\dd S+\mathcal B\,\dd P.
\label{eq:omegaAB}
\end{equation}
Its exterior derivative is
\begin{equation}
\dd\omega=
\left(\partial_S\mathcal B-\partial_P\mathcal A\right)
\dd S\wedge\dd P.
\end{equation}
Condition \eqref{eq:integrabilityAB} implies $\dd\omega=0$, so $\omega$ is closed. On a simply connected domain, a closed $C^1$ 1-form admits a global potential; hence
\begin{equation}
\omega=\dd\mathcal L.
\end{equation}
Equation \eqref{eq:ABgradient} follows. Because $\mathcal A,\mathcal B\in C^1$, the second derivatives of $\mathcal L$ are continuous and therefore $\mathcal L\in C^2(\Omega)$.

If $\mathcal L_1$ and $\mathcal L_2$ generate the same pair $(\mathcal A,\mathcal B)$, then $\dd(\mathcal L_1-\mathcal L_2)=0$; on a connected domain their difference is constant. The condition $\mathcal L(0,0)=0$ fixes that constant. Finally, the integral of an exact form depends only on its endpoints, proving path independence in Eq.~\eqref{eq:pathL}.
\end{proof}

If Eq.~\eqref{eq:ABparity} also holds, the reconstructed Lagrangian automatically inherits parity. Indeed, define $\widetilde{\mathcal L}(S,P)=\mathcal L(S,-P)$. Then
\begin{equation}
\widetilde{\mathcal L}_S
=\mathcal A(S,-P)=\mathcal A(S,P),
\end{equation}
and
\begin{equation}
\widetilde{\mathcal L}_P
=-\mathcal B(S,-P)=\mathcal B(S,P).
\end{equation}
By uniqueness in Theorem~\ref{thm:blind-reconstruction} and the same vacuum normalization,
\begin{equation}
\widetilde{\mathcal L}(S,P)=\mathcal L(S,P).
\end{equation}

\subsection{Causality directly in terms of first-order data}

The Hessian of the reconstructed Lagrangian can be expressed entirely in terms of $(\mathcal A,\mathcal B)$:
\begin{equation}
\nabla^2\mathcal L=
\begin{pmatrix}
\mathcal L_{SS} & \mathcal L_{SP}\\
\mathcal L_{PS} & \mathcal L_{PP}
\end{pmatrix}
=
\begin{pmatrix}
\partial_S\mathcal A & \partial_P\mathcal A\\
\partial_S\mathcal B & \partial_P\mathcal B
\end{pmatrix}.
\end{equation}
Using Eq.~\eqref{eq:integrabilityAB}, this becomes the symmetric matrix
\begin{equation}
H[\mathcal A,\mathcal B]=
\begin{pmatrix}
\partial_S\mathcal A & \partial_P\mathcal A\\
\partial_P\mathcal A & \partial_P\mathcal B
\end{pmatrix}.
\label{eq:HAB}
\end{equation}
Therefore, Eqs.~\eqref{eq:causal1}--\eqref{eq:causal2} are exactly equivalent to
\begin{equation}
\mathcal A>0,
\qquad
H[\mathcal A,\mathcal B]\succeq0.
\label{eq:ABweak}
\end{equation}

To write the remaining causality inequality compactly, introduce
\begin{equation}
\mathsf K(U,V,P)=
\begin{pmatrix}
2U & -P\\
-P & 2V
\end{pmatrix}.
\label{eq:Kmatrix}
\end{equation}
Since $U,V\ge0$ and $P^2=4UV$,
\begin{equation}
\det\mathsf K=4UV-P^2=0,
\qquad
\Tr\mathsf K=2(U+V)\ge0.
\end{equation}
Hence
\begin{equation}
\mathsf K\succeq0.
\end{equation}
For a nontrivial field, $U+V>0$, its eigenvalues are $0$ and $2(U+V)$, so
\begin{equation}
\operatorname{rank}\mathsf K=1.
\end{equation}
At the vacuum, $U=V=P=0$, one simply has $\mathsf K=0$ and rank zero.

Contracting $\mathsf K$ with the Hessian gives
\begin{align}
\Tr(\mathsf K H)
&=
2U\,\partial_S\mathcal A
+2V\,\partial_P\mathcal B
-2P\,\partial_P\mathcal A
\nonumber\\
&=
2U\mathcal L_{SS}
+2V\mathcal L_{PP}
-2P\mathcal L_{SP}.
\label{eq:TrKHexpanded}
\end{align}
Therefore, the strong condition \eqref{eq:causal3} becomes
\begin{equation}
\Tr(\mathsf K H)<\mathcal A.
\label{eq:ABstrong}
\end{equation}
All three causality restrictions are thus written solely in terms of the first-order constitutive data and their first derivatives.

To encode the Maxwell limit directly in these data, define $\varrho=(S^2+P^2)^{1/2}$ and require, as $\varrho\to0$,
\begin{equation}
\mathcal A(S,P)=1+O(\varrho),
\qquad
\mathcal B(S,P)=O(\varrho).
\label{eq:ABMaxwell}
\end{equation}
Integrating along a radial segment, Eq.~\eqref{eq:ABMaxwell} implies $\mathcal L(S,P)=S+O(\varrho^2)$ for the normalized Lagrangian. We define $\cCblind(\Omega)$ as the set of pairs $(\mathcal A,\mathcal B)\in C^1(\Omega)\times C^1(\Omega)$ satisfying Eqs.~\eqref{eq:ABparity}, \eqref{eq:integrabilityAB}, \eqref{eq:ABweak}, \eqref{eq:ABstrong}, and \eqref{eq:ABMaxwell}.

\begin{theorem}[Equivalence with the causal Pleba\'nski class]
\label{thm:blind-equivalence}
Let $\Omega\subset\R^2$ be an open simply connected domain that contains the vacuum and is invariant under $(S,P)\mapsto(S,-P)$. There is a one-to-one correspondence between parity-invariant, Maxwell-normalized, zero-vacuum-energy Lagrangians in $\cCcausal(\Omega)$ and constitutive pairs in $\cCblind(\Omega)$. In this sense,
\begin{equation}
\cCblind(\Omega)\simeq\cCcausal(\Omega)
\label{eq:blindEquiv}
\end{equation}
for the parity-invariant normalized subclass considered here.
\end{theorem}

\begin{proof}
Let first $\mathcal L\in\cCcausal(\Omega)$ be parity invariant and normalized. Defining
\begin{equation}
\mathcal A=\mathcal L_S,
\qquad
\mathcal B=\mathcal L_P,
\end{equation}
the Maxwell limit gives $\nabla\mathcal L(0,0)=(1,0)$ and, since the Hessian is locally bounded by $C^2$ regularity, $\nabla\mathcal L(S,P)=(1,0)+O(\varrho)$; this yields Eq.~\eqref{eq:ABMaxwell}. Parity of $\mathcal L$ yields Eq.~\eqref{eq:ABparity}, and equality of mixed derivatives yields Eq.~\eqref{eq:integrabilityAB}. Moreover, Eqs.~\eqref{eq:ABweak} and \eqref{eq:ABstrong} are precisely Eqs.~\eqref{eq:causal1}--\eqref{eq:causal3} expressed in terms of $(\mathcal A,\mathcal B)$. Hence every Lagrangian in the subclass considered defines an element of $\cCblind(\Omega)$.

Conversely, let $(\mathcal A,\mathcal B)\in\cCblind(\Omega)$. Theorem~\ref{thm:blind-reconstruction} reconstructs a unique $\mathcal L\in C^2(\Omega)$ with $\mathcal L(0,0)=0$. Condition \eqref{eq:ABparity} guarantees that $\mathcal L$ is even in $P$, while Eq.~\eqref{eq:ABMaxwell} guarantees the Maxwell limit. Finally, Eqs.~\eqref{eq:ABweak} and \eqref{eq:ABstrong} reproduce the necessary and sufficient causality conditions for the Pleba\'nski class \cite{Schellstede2016}. Thus $\mathcal L\in\cCcausal(\Omega)$, completing the one-to-one correspondence.
\end{proof}

The equivalence \eqref{eq:blindEquiv} combines two known ingredients---constitutive integrability \cite{AschieriFerrara2013} and Pleba\'nski causality \cite{Schellstede2016}---to formulate the inverse problem directly in the $(S,P)$ plane. It is not a new causal classification or a reformulation of the Poincar\'e lemma. Its usefulness is methodological: it avoids choosing a named Lagrangian or an input geometry and makes the functional freedom subject to integrability and causality inequalities explicit.

\subsection{Why the axioms do not select a unique Lagrangian}

The preceding freedom is not merely formal. If a theory lies strictly inside the causal region, then a full functional neighborhood of equally admissible theories exists.

\begin{theorem}[Functional stability under strict causal margins]
\label{thm:functional-stability}
Let $D\subset\R^2$ be a physical domain and let $\mathcal L_0\in\cCcausal(D)$ be parity invariant and Maxwell normalized. Let $\mathcal K\Subset\operatorname{int}D$ be a compact set with nonempty interior and suppose that, on $\mathcal K$, there are constants $m_1,m_2,m_3>0$ such that
\begin{equation}
\mathcal L_{0,S}\ge m_1,
\qquad
\lambda_{\min}(\nabla^2\mathcal L_0)\ge m_2,
\qquad
\mathscr C[\mathcal L_0]\ge m_3,
\label{eq:strictmargins}
\end{equation}
where
\begin{equation}
\mathscr C[\mathcal L]
=
\mathcal L_S
-2U\mathcal L_{SS}
-2V\mathcal L_{PP}
+2P\mathcal L_{SP}.
\label{eq:CfunctionalBlind}
\end{equation}
Let moreover
\begin{equation}
\psi\in C_c^2(\operatorname{int}\mathcal K)
\end{equation}
be even in $P$ and vanish in a neighborhood of the vacuum. Define
\begin{equation}
\mathcal L_\varepsilon
=\mathcal L_0+\varepsilon\psi.
\label{eq:functionaldeform}
\end{equation}
Then there exists $\varepsilon_*>0$ such that, for every $|\varepsilon|<\varepsilon_*$,
\begin{equation}
\mathcal L_\varepsilon\in\cCcausal(D),
\end{equation}
with the same parity, vacuum energy, and weak-field Maxwell limit as $\mathcal L_0$. If, in addition, $\psi$ vanishes throughout the sufficiently strong-field region, any self-energy condition depending only on the strong-field asymptotics is unchanged.
\end{theorem}

\begin{proof}
Because $\mathcal K$ is compact and $\psi\in C_c^2$, the quantities
\begin{equation}
\begin{aligned}
M_1&:=\|\psi_S\|_{L^\infty(\mathcal K)},\\
M_2&:=\|\nabla^2\psi\|_{\mathrm{op},L^\infty(\mathcal K)},\\
M_3&:=\|\mathscr C[\psi]\|_{L^\infty(\mathcal K)}
\end{aligned}
\label{eq:M123}
\end{equation}
are finite. We adopt the convention $m_i/M_i=+\infty$ when $M_i=0$ and define
\begin{equation}
\varepsilon_*=
\frac12\min\left\{
1,
\frac{m_1}{M_1},
\frac{m_2}{M_2},
\frac{m_3}{M_3}
\right\}>0.
\label{eq:epsilonstar}
\end{equation}

For $|\varepsilon|<\varepsilon_*$,
\begin{equation}
\mathcal L_{\varepsilon,S}
=\mathcal L_{0,S}+\varepsilon\psi_S
\ge
m_1-|\varepsilon|M_1>0.
\label{eq:stabA}
\end{equation}
Now let $v$ be a unit vector. On $\mathcal K$,
\begin{align}
v^T\nabla^2\mathcal L_\varepsilon v
&=
v^T\nabla^2\mathcal L_0v
+\varepsilon\,v^T\nabla^2\psi\,v
\nonumber\\
&\ge
m_2-|\varepsilon|M_2>0.
\label{eq:stabH}
\end{align}
Thus $\nabla^2\mathcal L_\varepsilon$ remains positive definite on $\mathcal K$. Finally, by linearity of $\mathscr C$,
\begin{equation}
\mathscr C[\mathcal L_\varepsilon]
=
\mathscr C[\mathcal L_0]
+\varepsilon\mathscr C[\psi]
\ge
m_3-|\varepsilon|M_3>0.
\label{eq:stabC}
\end{equation}
Hence all causality inequalities are preserved on $\mathcal K$.

Outside $\supp\psi$ one has
\begin{equation}
\mathcal L_\varepsilon=\mathcal L_0,
\end{equation}
so causality on the rest of $D$ is inherited directly from $\mathcal L_0$. Parity is preserved because both $\mathcal L_0$ and $\psi$ are even in $P$. Since $\psi$ vanishes in a neighborhood of the vacuum,
\begin{equation}
\mathcal L_\varepsilon=\mathcal L_0
\end{equation}
there, so neither $\mathcal L(0,0)=0$ nor the Maxwell normalization changes. If $\psi$ also vanishes for sufficiently strong fields, the deformed theory coincides exactly with $\mathcal L_0$ in the strong-field region and preserves any self-energy criterion controlled only by that asymptotic behavior.
\end{proof}

\begin{corollary}[Functional nonuniqueness]
\label{cor:functional-nonuniqueness}
Under the assumptions of Theorem~\ref{thm:functional-stability}, assume in addition that $\mathcal K$ is invariant under $(S,P)\mapsto(S,-P)$. Then a strictly causal Lagrangian is not an isolated point in the space of admissible theories: there exist infinitely many linearly independent even deformations $\psi_n\in C_c^2(\operatorname{int}\mathcal K)$, vanishing near the vacuum, and sufficiently small nonzero amplitudes for which $\mathcal L_0+\varepsilon_n\psi_n$ remains causal, parity invariant, and Maxwell normalized.
\end{corollary}

Thus the universal conditions leave infinitely many admissible local deformation directions. This statement is local and need not extend to real-analytic subclasses with globally fixed asymptotics; closed-form sectors therefore require an additional reduction.

\section{Representation in the \texorpdfstring{$(U,V)$}{(U,V)} plane}
\label{sec:UV}

The $(\mathcal A,\mathcal B)$ representation is the most direct way to describe the full class. To identify integrable substructures, it is useful to work in a simply connected region where the following quantities are positive:
\begin{equation}
p(U,V):=\mathcal L_V>0,
\qquad
q(U,V):=-\mathcal L_U>0.
\label{eq:pqdef}
\end{equation}
We do not assume that this patch covers all of $\cCcausal$ globally. Equality of mixed derivatives gives
\begin{equation}
p_U+q_V=0.
\label{eq:pqintegrability}
\end{equation}
Define
\begin{equation}
\Delta_d(U,V):=pq=-\mathcal L_U\mathcal L_V,
\qquad
R(U,V):=\frac{p}{q}.
\label{eq:DeltaR}
\end{equation}
Then
\begin{equation}
p=\sqrt{\Delta_dR},
\qquad
q=\sqrt{\frac{\Delta_d}{R}},
\end{equation}
and integrability is equivalent to
\begin{equation}
\partial_U\sqrt{\Delta_dR}
+\partial_V\sqrt{\frac{\Delta_d}{R}}=0.
\label{eq:DeltaRPDE}
\end{equation}

\begin{theorem}[Reconstruction in the $(\Delta_d,R)$ representation]
Let $\Delta_d,R>0$ be $C^1$ functions on a simply connected region and suppose they satisfy Eq.~\eqref{eq:DeltaRPDE}. Then
\begin{equation}
\omega=-\sqrt{\frac{\Delta_d}{R}}\,\dd U
+\sqrt{\Delta_dR}\,\dd V
\label{eq:omegaDeltaR}
\end{equation}
is closed, and there exists a Lagrangian $\mathcal L$ on that region with $\dd\mathcal L=\omega$, unique up to an additive constant.
\end{theorem}
\begin{proof}
The condition $\dd\omega=0$ is precisely Eq.~\eqref{eq:DeltaRPDE}; the result follows from the Poincar\'e lemma.
\end{proof}

Equivalently, if $\Delta_d$ is prescribed and one solves for $p$,
\begin{equation}
p_U+\partial_V\!\left(\frac{\Delta_d}{p}\right)=0,
\label{eq:pPDE}
\end{equation}
one obtains a first-order quasilinear equation. Its characteristics are recorded in Appendix~\ref{app:characteristics}. This representation remains blind: neither $\Delta_d$ nor $R$ has yet been fixed by a named theory. The causality inequalities must still be imposed afterward on each reconstructed solution.

\section{The known self-dual hypersurface}
\label{sec:selfdual}

Continuous electromagnetic self-duality corresponds to invariance of the equations of motion under electric--magnetic duality rotations in the Gaillard--Zumino--Gibbons--Rasheed framework \cite{GaillardZumino1981,GibbonsRasheed1995}. In the variables introduced above, the condition takes the form
\begin{equation}
\mathcal L_U\mathcal L_V=-1,
\qquad\text{i.e.}\qquad
\Delta_d=1.
\label{eq:selfdualDelta}
\end{equation}
The functional solution can be organized through the Courant--Hilbert construction \cite{CourantHilbert1962,RussoTownsendCausalSelfDual2024}:
\begin{align}
\mathcal L&=\ell(\tau)-\frac{2U}{\dot\ell(\tau)},
\label{eq:CHL}\\
\tau&=V+\frac{U}{\dot\ell(\tau)^2},
\label{eq:CHtau}
\end{align}
with
\begin{equation}
\mathcal L_V=\dot\ell,
\qquad
\mathcal L_U=-\dot\ell^{-1}.
\end{equation}
For the causal self-dual branch, the inequalities can be organized as
\begin{equation}
\dot\ell\ge1,
\qquad
\ddot\ell\ge0,
\label{eq:CHcausal}
\end{equation}
with the refinements discussed in the modern literature \cite{RussoTownsendCausalSelfDual2024,RussoTownsendDualities2024,RussoTownsendSimplified2025,KuzenkoRuhl2026}. To belong in addition to the strictly Maxwell-normalized class considered here, we impose
\begin{equation}
\ell(0)=0,
\qquad
\dot\ell(0)=1,
\label{eq:CHMaxwell}
\end{equation}
so that the vacuum energy vanishes and the linear term is exactly Maxwell's. This distinction separates our subclass from more general self-dual weak-field limits, including ModMax-type theories.

Equation~\eqref{eq:selfdualDelta} will be used only as a control against the known causal self-dual sector \cite{RussoTownsendCausalSelfDual2024,RussoTownsendSimplified2025,Chen2025,BabaeiAghbolagh2026,KuzenkoRuhl2026}. The next section introduces a distinct, non-self-dual reduction.

\section{An exactly integrable non-self-dual reduction}
\label{sec:CX}

\subsection{Integration of the constitutive ratio}

We now seek an analytic reduction that does not fix a Lagrangian in advance. To motivate the choice, consider first the separable class
\begin{equation}
R(U,V)=\frac{\alpha(U)}{\beta(V)},
\qquad
\alpha(0)=\beta(0)=1,
\label{eq:Rseparable}
\end{equation}
which is compatible with $R(0,0)=1$ in the Maxwell limit. Writing $p=\alpha h$ and $q=\beta h$, the integrability condition \eqref{eq:pqintegrability} becomes
\begin{equation}
\alpha h_U+\beta h_V+[\alpha'(U)+\beta'(V)]h=0.
\label{eq:separablecompat}
\end{equation}
The simplest separable reduction is the one for which compatibility becomes a homogeneous transport equation, with no source term proportional to $h$. Requiring
\begin{equation}
\alpha'(U)+\beta'(V)=0,
\qquad
\forall\,(U,V)
\label{eq:sourcefree}
\end{equation}
forces, by separation of variables,
\begin{equation}
\alpha'(U)=c,
\qquad
\beta'(V)=-c,
\end{equation}
with constant $c$. Vacuum normalization then gives $\alpha=1+cU$ and $\beta=1-cV$. For the branch $c>0$, introduce $T=2/c$ and obtain
\begin{equation}
R_T(U,V)=\frac{1+2U/T}{1-2V/T}.
\label{eq:RT}
\end{equation}
Thus Eq.~\eqref{eq:RT} is neither a consequence of causality nor claimed to be the only possible integrable reduction; it is the normalized affine-separable reduction that removes the source term from Eq.~\eqref{eq:separablecompat}. This gives a minimal and reproducible motivation for the first specific functional choice in the construction.

Define
\begin{equation}
\alpha_T=1+\frac{2U}{T},
\qquad
\beta_T=1-\frac{2V}{T},
\label{eq:alphabetaUV}
\end{equation}
and work on
\begin{equation}
D_T=\{U\ge0,\ 0\le V<T/2\},
\label{eq:DT}
\end{equation}
where $\alpha_T,\beta_T>0$. Differential inequalities are understood in the relative interior $U>0$, $0<V<T/2$ and extended continuously to the electric and magnetic axes whenever the Lagrangian is regular.

\begin{theorem}[Integration of the ratio $R_T$]
Every local solution of Eq.~\eqref{eq:pqintegrability} satisfying $p/q=R_T$ can be written as
\begin{equation}
p=\alpha_Tg(X),
\qquad
q=\beta_Tg(X),
\label{eq:pqg}
\end{equation}
with
\begin{equation}
X=\alpha_T\beta_T=1-\frac{2S}{T}-\frac{P^2}{T^2},
\label{eq:Xdef}
\end{equation}
and an arbitrary function $g(X)$. The corresponding Lagrangian is
\begin{equation}
\mathcal L_g(S,P)=-\frac{T}{2}\int_1^{X(S,P)}g(x)\,\dd x,
\label{eq:Lg}
\end{equation}
with $g(1)=1$ from Maxwell normalization.
\end{theorem}
\begin{proof}
From $p/q=\alpha_T/\beta_T$, write $p=\alpha_Th$ and $q=\beta_Th$. Substitution into Eq.~\eqref{eq:pqintegrability}, together with $(\alpha_T)_U=2/T$ and $(\beta_T)_V=-2/T$, gives
\begin{equation}
\alpha_Th_U+\beta_Th_V=0.
\end{equation}
The characteristics satisfy $\alpha_T\beta_T=\mathrm{const}$. Since
\[
\nabla X=\frac{2}{T}(\beta_T,-\alpha_T)\neq0
\]
throughout $D_T$, the level sets of $X$ locally label the characteristics; hence $h=g(X)$. Since $X_V=-2\alpha_T/T$ and $\mathcal L_V=\alpha_Tg(X)$, one finds $\mathcal L_X=-Tg/2$ and Eq.~\eqref{eq:Lg}.
\end{proof}

The combination $X$ was not postulated as a Lagrangian ansatz or as geometric input; it arises as a first integral of the compatibility condition specifically associated with Eq.~\eqref{eq:RT}. The distinction matters. Different constitutive ratios generally produce different first integrals and subclasses.

\subsection{Exact causal characterization}

Equation \eqref{eq:Lg} gives
\begin{align}
\mathcal L_S&=g,
\label{eq:LSg}\\
\mathcal L_{SS}&=-\frac{2}{T}g',
\label{eq:LSSg}\\
\mathcal L_{SP}&=-\frac{2P}{T^2}g',
\label{eq:LSPg}\\
\mathcal L_{PP}&=\frac{g}{T}-\frac{2P^2}{T^3}g'.
\label{eq:LPPg}
\end{align}
In particular,
\begin{equation}
\det\nabla^2\mathcal L_g
=-\frac{2}{T^2}gg'.
\label{eq:detHg}
\end{equation}

\begin{theorem}[Causal characterization of the $X$ class]
\label{thm:CXcausal}
Let $g\in C^1(0,\infty)$ with $g(1)=1$, and consider the Lagrangian \eqref{eq:Lg} on the full constitutive domain
\begin{equation}
D_T=\{U\ge0,\ 0\le V<T/2\},
\end{equation}
with no upper truncation in $U$. Then Eq.~\eqref{eq:Lg} satisfies the full causality conditions \eqref{eq:causal1}--\eqref{eq:causal3} throughout $D_T$ if and only if
\begin{equation}
g(X)>0,
\qquad
g'(X)\le0,
\qquad
g(X)+2Xg'(X)\ge0.
\label{eq:gcausal}
\end{equation}
\end{theorem}
\begin{proof}
The first two conditions in Eq.~\eqref{eq:gcausal}, together with Eqs.~\eqref{eq:LPPg} and \eqref{eq:detHg}, are equivalent to $\mathcal L_S>0$ and positive semidefiniteness of the Hessian. For the strong condition, note that
\begin{equation}
\mathcal L_U=-\beta_Tg,
\qquad
\mathcal L_{UU}=-\frac{2\beta_T^2}{T}g'.
\end{equation}
Hence
\begin{equation}
\mathcal L_U+2U\mathcal L_{UU}
=-\beta_T\,[g+2(X-\beta_T)g'].
\label{eq:strongg}
\end{equation}
If Eq.~\eqref{eq:gcausal} holds,
\begin{equation}
g+2(X-\beta_T)g'=(g+2Xg')-2\beta_Tg'>0,
\end{equation}
because $g>0$ prevents the relevant terms from vanishing simultaneously. Conversely, at fixed $X>0$ one may take $\beta_T\to0^+$ while keeping $\alpha_T=X/\beta_T$; Eq.~\eqref{eq:strongg} then requires $g+2Xg'\ge0$.
\end{proof}

\begin{remark}
The necessity of $g+2Xg'\ge0$ is a global statement on $D_T$. The preceding step uses configurations with fixed $X>0$ and $\beta_T\to0^+$, for which $\alpha_T=X/\beta_T$ and hence $U$ can grow without bound. On a constitutive domain truncated at finite $U$, that limit would not by itself establish necessity of the third inequality.
\end{remark}

Introduce the local logarithmic index
\begin{equation}
a(X):=-X\frac{g'(X)}{g(X)}.
\label{eq:adef}
\end{equation}

\begin{corollary}[Functional parametrization of $\cCX$]
For positive $g\in C^1$, conditions \eqref{eq:gcausal} are equivalent to introducing a continuous function $a\in C^0(0,\infty)$ such that
\begin{equation}
{\displaystyle 0\le a(X)\le\frac12,}
\label{eq:abound}
\end{equation}
and, with $g(1)=1$,
\begin{equation}
g(X)=\exp\!\left[-\int_1^X\frac{a(y)}{y}\,\dd y\right].
\label{eq:gofa}
\end{equation}
Therefore,
\begin{equation}
{\displaystyle 
\mathcal L_a(S,P)=-\frac{T}{2}\int_1^X
\exp\!\left[-\int_1^x\frac{a(y)}{y}\,\dd y\right]\dd x
}
\label{eq:La}
\end{equation}
defines, for every continuous $a$ satisfying Eq.~\eqref{eq:abound}, an infinite-dimensional functional class of causal NLED theories within the reduction \eqref{eq:RT}.
\end{corollary}

The class $\cCX$ does not exhaust $\cCcausal$: it solves causality exactly only within the chosen reduction and on $D_T$. Its usefulness is that the remaining freedom is encoded by a single bounded function $a(X)$ rather than by a discrete selection of models.

\subsection{Self-duality and recovery of Born--Infeld}

For Eq.~\eqref{eq:Lg},
\begin{equation}
-\mathcal L_U=\beta_Tg,
\qquad
\mathcal L_V=\alpha_Tg,
\end{equation}
so
\begin{equation}
\Delta_d=-\mathcal L_U\mathcal L_V=Xg(X)^2.
\label{eq:DeltadX}
\end{equation}

\begin{theorem}[Intersection with the self-dual branch]
Within $\cCX$, the condition $\Delta_d=1$ uniquely selects
\begin{equation}
g(X)=X^{-1/2},
\qquad
a(X)=\frac12,
\end{equation}
and
\begin{equation}
\mathcal L_{\rm BI}=T(1-\sqrt X).
\label{eq:BI}
\end{equation}
Consequently,
\begin{equation}
{\displaystyle \cCX\cap\cCSD=\{\text{Born--Infeld}\}.}
\end{equation}
\end{theorem}
\begin{proof}
Self-duality requires $Xg^2=1$. Since $g>0$, one has $g=X^{-1/2}$; integrating Eq.~\eqref{eq:Lg} yields Eq.~\eqref{eq:BI}. Conversely, for $g=X^{-1/2}$ one has $g>0$, $g'<0$, $g+2Xg'=0$, and $\Delta_d=Xg^2=1$, so Born--Infeld belongs to both classes.
\end{proof}

This result is consistent with the exceptional role of Born--Infeld in the nonbirefringent classification \cite{RussoTownsendNoBirefringence2023,MezincescuRussoTownsend2024,GibbonsHerdeiro2001}.

\subsection{Known power-law subfamily}

If $a(X)=\gamma$ is constant,
\begin{equation}
g(X)=X^{-\gamma},
\qquad
0\le\gamma\le\frac12.
\end{equation}
Defining $q=1-\gamma$, we obtain on $D_T$
\begin{equation}
\mathcal L^{(q)}=\frac{T}{2q}(1-X^q),
\qquad
\frac12\le q\le1.
\label{eq:Lpower}
\end{equation}
For $1/2\le q<1$ this reproduces the known causal power family \cite{RussoTownsendBornAgain2024,RussoTownsendBH2026}; $q=1/2$ is Born--Infeld. The endpoint $q=1$ requires a distinction between the constitutive domain and the maximal extension: in our construction it gives
\begin{equation}
\mathcal L^{(1)}=S+\frac{P^2}{2T},
\end{equation}
which satisfies Eqs.~\eqref{eq:causal1}--\eqref{eq:causal3} only on the patch $D_T$, where $2V<T$. As a polynomial, however, it extends beyond that patch and fails the strong-field causality test there; for this reason it does not belong to the globally causal family $1/2\le q<1$ identified in Ref.~\cite{RussoTownsendBornAgain2024}. The contribution of Eq.~\eqref{eq:La} is not to rediscover the power-law line, but to replace the constant index by a function $a(X)$ whose causal region within the reduction is known exactly.

As a weak-field benchmark, the same branch admits a direct match to the one-loop effective action of quantum electrodynamics (QED). In our conventions, the one-loop Euler--Heisenberg expansion through quartic order can be written as \cite{HeisenbergEuler1936,Dunne2012}
\begin{align}
\mathcal L_{\rm EH}
&=S+\mu_{\rm EH}\left(4S^2+7P^2\right)+O(F^6),
\label{eq:EHweak}\\
\mu_{\rm EH}
&=\frac{2\alpha_{\rm em}^2}{45m_e^4}.
\label{eq:muEH}
\end{align}
On the other hand, expanding the $\cCX$ class about the Maxwell vacuum gives
\begin{equation}
\mathcal L_g
=S+\frac{a(1)}{T}S^2+\frac{1}{2T}P^2+O(F^6).
\label{eq:CXweakEH}
\end{equation}
Matching the coefficients in Eqs.~\eqref{eq:EHweak} and \eqref{eq:CXweakEH} fixes
\begin{equation}
\begin{aligned}
a(1)&=\frac{2}{7},
&\qquad T&=\frac{1}{14\mu_{\rm EH}},\\
q_{\rm EH}&=\frac{5}{7},
&\qquad &(a=\gamma=1-q).
\end{aligned}
\label{eq:EHmatch}
\end{equation}
Thus the Euler--Heisenberg matching point of the constant-index branch lies strictly inside the causal interval $1/2<q<1$. The match is a weak-field statement through $O(F^4)$ and does not identify the full Euler--Heisenberg effective action with a member of $\cCX$.

\subsection{Continuous mixtures of Born--Infeld}

Convexity of $\cCcausal$ allows one to construct, on a common reality domain, the mixture
\begin{equation}
\mathcal L^{\rm MBI}_\mu(S,P)=
\int_0^\infty
\left[T-\sqrt{T^2-2TS-P^2}\right]\mu(\dd T),
\label{eq:MBI}
\end{equation}
where $\mu$ is a positive probability measure and we assume
\begin{equation}
T_0:=\inf\supp\mu>0.
\label{eq:MBIT0}
\end{equation}
This hypothesis guarantees a nontrivial common open domain,
\begin{equation}
D_\mu=\{U\ge0,\ 0\le V<T_0/2\},
\label{eq:Dmu}
\end{equation}
on which all seeds are real and causal. Under the same domination assumptions used in Eq.~\eqref{eq:convexIntegral}, the mixture belongs to $\cCcausal(D_\mu)$. Without a positive lower bound on the support, the intersection of the reality domains may collapse onto the axis $V=0$, and we would not claim two-invariant causality on a common open neighborhood. The mixture is nevertheless generically non-self-dual. Defining
\begin{equation}
b_T=\sqrt{\frac{T+2U}{T-2V}},
\end{equation}
one finds
\begin{equation}
-\mathcal L^{\rm MBI}_{\mu,U}\mathcal L^{\rm MBI}_{\mu,V}
=\avg{b_T}_\mu\avg{b_T^{-1}}_\mu\ge1
\label{eq:MBIcs}
\end{equation}
by Cauchy--Schwarz. Equality requires $b_T$ to be constant $\mu$-almost surely; a nondegenerate measure therefore breaks continuous self-duality on a nontrivial background. We state this construction with cautious priority claims because convex combinations of nonlinear theories have precedents \cite{LiuYang2025}.

\subsection{A nonconstant member: \texorpdfstring{$\mathcal L_\star$}{L star}}

To demonstrate constructively that $\cCX$ contains closed-form solutions beyond the power-law line, choose
\begin{equation}
a_\star(X)=\frac{1+2X}{4(1+X)},
\label{eq:astar}
\end{equation}
for which $1/4<a_\star<1/2$ on $X>0$. Equation \eqref{eq:gofa} integrates to
\begin{equation}
g_\star(X)=\left[\frac{2}{X(1+X)}\right]^{1/4}.
\label{eq:gstar}
\end{equation}
Moreover,
\begin{equation}
g_\star+2Xg_\star'
=\frac{g_\star}{2(1+X)}>0,
\end{equation}
so the theory lies strictly inside the causal region except at the relevant limiting configurations.

A convenient antiderivative is
\begin{equation}
\Phi(X)=\frac{4\,2^{1/4}}{3}X^{3/4}
{}_2F_1\!\left(\frac14,\frac34;\frac74;-X\right),
\label{eq:Phistar}
\end{equation}
with $\Phi'=g_\star$, and
\begin{equation}
{\displaystyle \mathcal L_\star(S,P)=\frac{T}{2}[\Phi(1)-\Phi(X)].}
\label{eq:Lstar}
\end{equation}
Its self-duality defect is
\begin{equation}
\Delta_{d,\star}=Xg_\star^2
=\sqrt{\frac{2X}{1+X}}.
\label{eq:Deltastar}
\end{equation}
For $0<X<1$, $\Delta_{d,\star}<1$. By contrast, every positive mixture of self-dual seeds satisfies an inequality of the form \eqref{eq:MBIcs}. Therefore $\mathcal L_\star$ does not belong to the positive convex hull of Born--Infeld seeds considered in Eq.~\eqref{eq:MBI}: it is a genuinely non-self-dual direction within $\cCX$.

\section{Magnetic projection and spectral bootstrap}
\label{sec:spectral}

The class $\cCX$ is a full two-invariant theory. Complementarily, there is a one-variable inverse problem on the purely magnetic axis. This projection is weaker than full Pleba\'nski causality, and we keep that distinction explicit.

On a magnetic background,
\begin{equation}
P=0,
\qquad
S=-\Fcal,
\qquad
\Fcal=\frac{B^2}{2}>0.
\end{equation}
Define
\begin{equation}
\Lmag(\Fcal):=-\mathcal L(-\Fcal,0),
\qquad
A(\Fcal):=\Lmag_{\Fcal}>0.
\label{eq:magdefs}
\end{equation}
For an effective one-invariant theory on this axis, the extraordinary characteristic factor is \cite{DeLorenci2000,Schellstede2016,FathiGuzmanVillanueva2026}
\begin{equation}
\kappa(\Fcal)=1+2\Fcal\frac{A'(\Fcal)}{A(\Fcal)}.
\label{eq:kappa}
\end{equation}
We impose
\begin{equation}
A>0,
\qquad
0<\kappa\le1.
\label{eq:magcausal}
\end{equation}
These are magnetic-axis conditions; by themselves they do not establish full causality for an arbitrary extension $\mathcal L(S,P)$.

\begin{theorem}[Inverse reconstruction from the magnetic cone]
Let $\kappa:(0,\infty)\to(0,1]$ be continuous and suppose that, for some $\Fcal_0>0$,
\begin{equation}
\int_0^{\Fcal_0}\frac{|1-\kappa(u)|}{u}\,\dd u<\infty.
\label{eq:kappaint}
\end{equation}
Then there exists a unique response $A\in C^1(0,\infty)$, continuous at the origin and Maxwell normalized by $A(0)=1$, with characteristic factor \eqref{eq:kappa}:
\begin{equation}
A(\Fcal)=\exp\!\left[
\frac12\int_0^{\Fcal}\frac{\kappa(u)-1}{u}\,\dd u
\right],
\label{eq:Afromkappa}
\end{equation}
\begin{equation}
\Lmag(\Fcal)=\int_0^{\Fcal}A(v)\,\dd v.
\label{eq:LfromA}
\end{equation}
\end{theorem}
\begin{proof}
Equation \eqref{eq:kappa} is equivalent to
\begin{equation}
\frac{\dd\ln A}{\dd\ln\Fcal}=\frac{\kappa-1}{2}.
\end{equation}
Assumption \eqref{eq:kappaint} allows the integration constant to be fixed by $A(0)=1$; Eq.~\eqref{eq:LfromA} fixes $\Lmag(0)=0$.
\end{proof}

\begin{remark}
Local integrability of $\kappa$ alone would reconstruct a locally absolutely continuous response and make Eq.~\eqref{eq:kappa} valid almost everywhere. Continuity is assumed here because the subsequent characteristic and Einstein equations are used in the classical sense.
\end{remark}

\subsection{Positive spectral representation}

We now introduce an additional structural assumption: complete monotonicity,
\begin{equation}
(-1)^nA^{(n)}(\Fcal)\ge0,
\qquad n=0,1,2,\ldots.
\label{eq:CM}
\end{equation}
This condition is stronger than magnetic causality and is not implied by Eq.~\eqref{eq:magcausal}. Its role is structural: it represents the response as a positive superposition of exponential kernels and makes closure under positive mixing, moment estimates, and support restrictions available without assigning a microscopic interpretation to the spectral variable. The Bernstein--Widder theorem guarantees a unique positive probability measure $\rho$ on $[0,\infty)$ \cite{Widder1941,Schilling2012} such that
\begin{equation}
A(\Fcal)=\int_0^\infty\e^{-s\Fcal}\rho(\dd s),
\label{eq:LaplaceA}
\end{equation}
and
\begin{equation}
\Lmag(\Fcal)=
\int_0^\infty\frac{1-\e^{-s\Fcal}}{s}\rho(\dd s),
\label{eq:LaplaceL}
\end{equation}
where the integrand at $s=0$ is understood by continuity. We denote by $\cCrho$ the class of Maxwell-normalized magnetic responses admitting a positive representation of this form.

In all self-energy statements below we assume a nonvanishing magnetic charge,
\begin{equation}
Q_m\neq0,
\qquad
C=\frac{Q_m^2}{2}>0.
\label{eq:nonzeroQm}
\end{equation}

\begin{theorem}[Exact $-1/4$ moment criterion]
\label{thm:minusquarter}
For a magnetic monopole with
\begin{equation}
\Fcal(r)=\frac{C}{r^4},
\qquad
C=\frac{Q_m^2}{2},
\label{eq:Fmagnetic}
\end{equation}
the reduced self-energy
\begin{equation}
\Ered=\int_0^\infty r^2\Lmag(C/r^4)\,\dd r
\end{equation}
is finite if and only if
\begin{equation}
{\displaystyle 
\int_0^\infty s^{-1/4}\rho(\dd s)<\infty,
}
\label{eq:minusquarter}
\end{equation}
and, when it converges,
\begin{equation}
\Ered=
\frac{\Gamma(1/4)}{3}C^{3/4}
\int_0^\infty s^{-1/4}\rho(\dd s).
\label{eq:Eredmoment}
\end{equation}
\end{theorem}
\begin{proof}
By Tonelli's theorem, using positivity,
\begin{equation}
\Ered=
\int_0^\infty\frac{\rho(\dd s)}{s}
\int_0^\infty r^2\bigl(1-\e^{-sC/r^4}\bigr)\,\dd r.
\end{equation}
The substitution $y=sC/r^4$ gives
\begin{equation}
\int_0^\infty r^2(1-\e^{-a/r^4})\,\dd r
=\frac{\Gamma(1/4)}{3}a^{3/4},
\end{equation}
which yields Eqs.~\eqref{eq:minusquarter} and \eqref{eq:Eredmoment}.
\end{proof}

\begin{theorem}[Absence of a positive spectral gap]
Suppose Eq.~\eqref{eq:LaplaceA} satisfies global magnetic causality \eqref{eq:magcausal} and has finite self-energy. Then
\begin{equation}
{\displaystyle 
\inf\supp\rho=0,
\qquad
\rho(\{0\})=0.
}
\label{eq:gapless}
\end{equation}
In other words, the measure cannot have a gap between the origin and the beginning of its support. In particular, no positive measure with finitely many atoms is admissible.
\end{theorem}
\begin{proof}
Define the tilted measure
\begin{equation}
\rho_{\Fcal}(\dd s)=
\frac{\e^{-s\Fcal}\rho(\dd s)}{A(\Fcal)}.
\label{eq:tiltedrho}
\end{equation}
Then
\begin{equation}
\kappa(\Fcal)=1-2\Fcal\avg{s}_{\Fcal}.
\label{eq:kappameanrho}
\end{equation}
If $s_0=\inf\supp\rho>0$, then $\avg{s}_{\Fcal}\ge s_0$ and $\kappa\le1-2s_0\Fcal$, which becomes negative at sufficiently strong field. Hence $\inf\supp\rho=0$. If $\rho(\{0\})=w_0>0$, the atom at the origin contributes exactly $w_0\Fcal$ and, since all remaining contributions are nonnegative, $\Lmag(\Fcal)\ge w_0\Fcal$. With Eq.~\eqref{eq:Fmagnetic}, $r^2\Lmag(C/r^4)\ge w_0C/r^2$, so the self-energy diverges. A measure with finitely many atoms has either a strictly positive smallest atom or an atom at zero and is excluded in either case.
\end{proof}

\subsection{Exact lifting criterion to the \texorpdfstring{$X$}{X} class}

On the magnetic axis of $\cCX$,
\begin{equation}
X=1+\frac{2\Fcal}{T},
\qquad
A(\Fcal)=g\!\left(1+\frac{2\Fcal}{T}\right).
\end{equation}

\begin{proposition}[$X$-lifting criterion]
Let $A\in C^1([0,\infty))$ with $A(0)=1$. For a fixed scale $T>0$, the magnetic law $A(\Fcal)$ admits a completion of the form \eqref{eq:Lg} satisfying Theorem~\ref{thm:CXcausal} if and only if
\begin{equation}
A>0,
\qquad
A'\le0,
\qquad
A+(T+2\Fcal)A'\ge0.
\label{eq:Xlift}
\end{equation}
\end{proposition}
\begin{proof}
Necessity follows from $A(\Fcal)=g(X)$ and $X=1+2\Fcal/T$, since $g'=TA'/2$ and $2Xg'=(T+2\Fcal)A'$. Conversely, for $X\ge1$ define
\[
g_+(X)=A\!\left[\frac{T}{2}(X-1)\right].
\]
Equation~\eqref{eq:Xlift} gives the three inequalities in \eqref{eq:gcausal}. Let $a_1=-TA'(0)/2$; evaluating \eqref{eq:Xlift} at $\Fcal=0$ gives $0\le a_1\le1/2$. For $0<X\le1$ set $g_-(X)=X^{-a_1}$. Then $g_->0$, $g_-'\le0$, and $g_-+2Xg_-'=X^{-a_1}(1-2a_1)\ge0$, while $g_-(1)=g_+(1)$ and $g_-'(1)=g_+'(1)$. The two branches therefore define a $C^1(0,\infty)$ function causal on all $X>0$ whose magnetic restriction is precisely $A$.
\end{proof}

The last inequality in Eq.~\eqref{eq:Xlift} is stronger than $A+2\Fcal A'>0$ when $A'\le0$, so $\cCrho$ is not contained in $\cCX$ in general. For
\begin{equation}
A=(1+2\Fcal/T)^{-\gamma},
\end{equation}
by contrast,
\begin{equation}
A+(T+2\Fcal)A'=(1-2\gamma)A,
\end{equation}
and every $0\le\gamma\le1/2$ lifts to the power subfamily \eqref{eq:Lpower}.

The lifting criterion also provides a direct bridge between the non-power member of $\cCX$ and the spectral class. On its magnetic axis, Eq.~\eqref{eq:gstar} gives
\begin{equation}
A_\star(\Fcal)=
\left(1+\frac{2\Fcal}{T}\right)^{-1/4}
\left(1+\frac{\Fcal}{T}\right)^{-1/4}.
\label{eq:Astarfactor}
\end{equation}
Each factor is the Laplace transform of a gamma law. Hence $A_\star$ is completely monotone and belongs simultaneously to $\cCX$ and $\cCrho$. More explicitly,
\begin{equation}
A_\star(\Fcal)=\int_0^\infty \e^{-s\Fcal}\rho_\star(s)\,\dd s,
\label{eq:Astarspectral}
\end{equation}
with
\begin{equation}
{\displaystyle 
\rho_\star(s)=
\frac{T^{1/2}}{2^{1/4}\sqrt\pi}
s^{-1/2}\e^{-Ts}
{}_1F_1\!\left(\frac14;\frac12;\frac{Ts}{2}\right).}
\label{eq:rhostar}
\end{equation}
Indeed, $\rho_\star$ is the convolution of gamma densities with common shape $1/4$ and rates $T/2$ and $T$. Its behavior $\rho_\star(s)=O(s^{-1/2})$ at the origin also makes the moment in Eq.~\eqref{eq:minusquarter} finite. Thus the two-invariant and spectral constructions have a nontrivial intersection beyond the power-law line.

\subsection{Generalized-gamma family}

The generalized-gamma distribution was introduced by Stacy \cite{Stacy1962}. We use it here as a spectral measure:
\begin{equation}
\rho_{\gamma,\eta,\beta}(s)=
\frac{\eta}{\beta^\gamma\Gamma(\gamma/\eta)}
 s^{\gamma-1}
\exp\!\left[-\left(\frac{s}{\beta}\right)^\eta\right],
\label{eq:GGdensity}
\end{equation}
with $\gamma,\eta,\beta>0$. Defining
\begin{equation}
x=\beta\Fcal,
\end{equation}
and changing variables according to $s=\beta t$, the response becomes
\begin{equation}
A_{\gamma,\eta}(x)=
\frac{\eta}{\Gamma(\gamma/\eta)}
\int_0^\infty t^{\gamma-1}\e^{-t^\eta-xt}\,\dd t.
\label{eq:AGG}
\end{equation}
By construction, $A_{\gamma,\eta}$ is completely monotone.

The moment in Theorem~\ref{thm:minusquarter} evaluates exactly to
\begin{equation}
\int_0^\infty s^{-1/4}\rho_{\gamma,\eta,\beta}(s)\,\dd s
=\beta^{-1/4}
\frac{\Gamma[(\gamma-1/4)/\eta]}{\Gamma(\gamma/\eta)}.
\label{eq:GGmoment}
\end{equation}
Therefore, the self-energy is finite if and only if
\begin{equation}
\gamma>\frac14.
\label{eq:gammaenergy}
\end{equation}

The point that cannot be left to a purely asymptotic argument is global magnetic causality. For this family it can be settled exactly.

\begin{theorem}[Global magnetic causality of the generalized-gamma family]
For the response \eqref{eq:AGG}, the factor \eqref{eq:kappa} satisfies
\begin{equation}
0<\kappa_{\gamma,\eta}(x)\le1,
\qquad
\forall\,x\ge0
\label{eq:kappaallx}
\end{equation}
if and only if
\begin{equation}
0<\gamma\le\frac12,
\label{eq:gammaCausal}
\end{equation}
independently of $\eta>0$.
\end{theorem}
\begin{proof}
Define
\begin{equation}
Z(x)=\int_0^\infty t^{\gamma-1}\e^{-t^\eta-xt}\,\dd t
\end{equation}
and the tilted probability measure
\begin{equation}
\dd\pi_x(t)=\frac{t^{\gamma-1}\e^{-t^\eta-xt}}{Z(x)}\,\dd t.
\label{eq:pix}
\end{equation}
Since $A'/A=-\avg{t}_x$,
\begin{equation}
\kappa=1-2x\avg{t}_x.
\label{eq:kappaMean}
\end{equation}
Integration by parts gives
\begin{equation}
0=\int_0^\infty
\frac{\dd}{\dd t}
\left[t^\gamma\e^{-t^\eta-xt}\right]\dd t,
\end{equation}
which yields the identity
\begin{equation}
x\avg{t}_x+\eta\avg{t^\eta}_x=\gamma.
\label{eq:virialGG}
\end{equation}
Therefore,
\begin{equation}
\kappa_{\gamma,\eta}(x)
=1-2\gamma+2\eta\avg{t^\eta}_x.
\label{eq:kappaGGexact}
\end{equation}
Since $A'\le0$, Eq.~\eqref{eq:kappaMean} gives $\kappa\le1$. If $0<\gamma<1/2$, Eq.~\eqref{eq:kappaGGexact} implies $\kappa>1-2\gamma>0$. If $\gamma=1/2$, strict positivity of $\avg{t^\eta}_x$ for every finite $x$ gives $\kappa>0$, with the marginal limit $\kappa\to0$ as $x\to\infty$. To prove necessity for $\gamma>1/2$, rescale $t=y/x$ in Eq.~\eqref{eq:AGG}; by dominated convergence, equivalently by Watson's lemma,
\begin{equation}
x\avg{t}_x\longrightarrow\gamma,
\qquad x\to\infty,
\end{equation}
so $\kappa\to1-2\gamma<0$.
\end{proof}

Magnetic causality does not by itself imply a causal two-invariant completion. For generalized-gamma spectra the lifting question can nevertheless be answered within $\cCX$. Let
\begin{equation}
m_{\gamma,\eta}(x):=-\frac{A'_{\gamma,\eta}(x)}{A_{\gamma,\eta}(x)}
=\avg{t}_x>0
\end{equation}
and define
\begin{equation}
R_{\gamma,\eta}(x):=
\frac{1}{m_{\gamma,\eta}(x)}-2x,
\qquad
\tau_{\gamma,\eta}:=\inf_{x\ge0}R_{\gamma,\eta}(x).
\label{eq:taulift}
\end{equation}

\begin{theorem}[Causal completion of generalized-gamma responses]
\label{thm:GGlift}
Let $0<\gamma\le1/2$, $\eta>0$, and set $x=\beta\Fcal$. The magnetic response $A_{\gamma,\eta}(\beta\Fcal)$ admits a completion in $\cCX$ with scale $T>0$ if and only if
\begin{equation}
0<\beta T\le\tau_{\gamma,\eta}.
\label{eq:GGliftcondition}
\end{equation}
Moreover,
\begin{enumerate}[label=(\roman*)]
\item $\tau_{\gamma,\eta}>0$ for $0<\gamma<1/2$ and every $\eta>0$;
\item $\tau_{1/2,\eta}>0$ for $0<\eta\le1$, with $\tau_{1/2,1}=2$;
\item $\tau_{1/2,\eta}=0$ for $\eta>1$, so no completion of the form \eqref{eq:Lg} exists at this boundary.
\end{enumerate}
\end{theorem}
\begin{proof}
Because $A'_{\gamma,\eta}/A_{\gamma,\eta}=-m_{\gamma,\eta}$, the last inequality in Eq.~\eqref{eq:Xlift} becomes
\begin{equation}
1-(\beta T+2x)m_{\gamma,\eta}(x)\ge0,
\end{equation}
which is equivalent to Eq.~\eqref{eq:GGliftcondition}. Global magnetic causality gives
\begin{equation}
R_{\gamma,\eta}(x)=
\frac{\kappa_{\gamma,\eta}(x)}{m_{\gamma,\eta}(x)}>0
\end{equation}
at every finite $x$. Expanding the defining Laplace integral at large $x$ yields
\begin{equation}
m_{\gamma,\eta}(x)=
\frac{\gamma}{x}
-\eta\frac{\Gamma(\gamma+\eta)}{\Gamma(\gamma)}
x^{-\eta-1}+o(x^{-\eta-1}).
\label{eq:mGGasym}
\end{equation}
For $\gamma<1/2$,
\begin{equation}
R_{\gamma,\eta}(x)=
\left(\frac1\gamma-2\right)x+o(x)\longrightarrow\infty.
\end{equation}
For $\gamma=1/2$,
\begin{equation}
R_{1/2,\eta}(x)\sim
4\eta\frac{\Gamma(\eta+1/2)}{\Gamma(1/2)}x^{1-\eta}.
\label{eq:Rboundary}
\end{equation}
It tends to infinity for $\eta<1$, equals $2$ identically for $\eta=1$, and tends to zero for $\eta>1$. Since $R$ is continuous and strictly positive on every compact interval, these limits prove all three statements.
\end{proof}

\begin{corollary}[Exact physical window]
Within the generalized-gamma family, global magnetic causality and finite self-energy are simultaneously equivalent to
\begin{equation}
{\displaystyle 
\frac14<\gamma\le\frac12,
\qquad
\eta>0,
\qquad
\beta>0.
}
\label{eq:GGwindow}
\end{equation}
\end{corollary}

Condition \eqref{eq:GGwindow} cleanly separates the roles of the parameters: $\gamma$ fixes the physical boundaries, whereas $\eta$ deforms the measure and the finite-field interpolation without moving them. The endpoints $1/4$ and $1/2$ already appear on the branch $\eta=1$ and in Ref.~\cite{FathiGuzmanVillanueva2026}; the new result is their exact validity throughout the generalized-gamma family for all $\eta>0$. Appendix~\ref{app:spectralproofs} connects this property with the universality of regularly varying laws \cite{BinghamGoldieTeugels1987,RussoTownsendBH2026,FathiGuzmanVillanueva2026}.

Theorem~\ref{thm:GGlift} sharpens the scope of this statement. The full window \eqref{eq:GGwindow} is magnetically causal, but only the subregions specified there are guaranteed to admit a causal two-invariant completion within the affine-separable class $\cCX$.

\subsection{Mellin--Barnes representation and hierarchy of special functions}

Using the classical Mellin--Barnes inversion framework \cite{Barnes1908,ParisKaminski2001,DubovykGluzaSomogyi2022} for $\e^{-xt}$ with $0<c<\gamma$,
\begin{equation}
\e^{-xt}=\frac{1}{2\pi i}
\int_{c-i\infty}^{c+i\infty}\Gamma(z)(xt)^{-z}\,\dd z,
\end{equation}
Eq.~\eqref{eq:AGG} gives
\begin{equation}
A_{\gamma,\eta}(x)=
\frac{1}{\Gamma(\gamma/\eta)}
\frac{1}{2\pi i}
\int_{c-i\infty}^{c+i\infty}
\Gamma(z)\Gamma\!\left(\frac{\gamma-z}{\eta}\right)
 x^{-z}\,\dd z.
\label{eq:AMB}
\end{equation}
This is an exact representation for every $\eta>0$ in its Mellin--Barnes domain. Mellin--Barnes integrals with this gamma-factor structure are closely connected with the method of brackets, originally developed as a heuristic procedure for definite integrals and later applied directly to Mellin--Barnes and inverse Mellin representations \cite{GonzalezMoll2010,GonzalezKondrashukMollRecabarren2022,AnanthanarayanBanikFriotPathak2023}. We do not use the bracket rules in the present derivation: Eq.~\eqref{eq:AMB} follows directly from Mellin inversion. The comparison is nevertheless relevant because it provides an alternative route from the same contour-integral structure to hypergeometric series and clarifies the relation between bracket expansions and residue-based Mellin--Barnes evaluation. With the standard Fox--$H$ convention and its classical analytic-continuation theory \cite{Fox1961,Braaksma1964,MathaiSaxenaHaubold2010,KilbasSaigo1998,Giraldi2025},
\begin{equation}
A_{\gamma,\eta}(x)=
\frac{1}{\Gamma(\gamma/\eta)}
H_{1,1}^{1,1}\!\left[
 x\,\middle|\,
\begin{matrix}
(1-\gamma/\eta,1/\eta)\\
(0,1)
\end{matrix}
\right].
\label{eq:AFoxH}
\end{equation}
Integrating from the vacuum and taking $0<c<\min(\gamma,1)$,
\begin{equation}
\Lmag_{\gamma,\eta,\beta}(\Fcal)=
\frac{x}{\beta\Gamma(\gamma/\eta)}
\frac{1}{2\pi i}\int
\frac{\Gamma(z)\Gamma[(\gamma-z)/\eta]}{1-z}
 x^{-z}\,\dd z,
\label{eq:LMB}
\end{equation}
which is equivalent to
\begin{equation}
\Lmag_{\gamma,\eta,\beta}(\Fcal)=
\frac{x}{\beta\Gamma(\gamma/\eta)}
H_{2,2}^{1,2}\!\left[
 x\,\middle|\,
\begin{matrix}
(1-\gamma/\eta,1/\eta),(0,1)\\
(0,1),(-1,1)
\end{matrix}
\right].
\label{eq:LFoxH}
\end{equation}

The moment expansion is
\begin{equation}
A_{\gamma,\eta}(x)
\sim\sum_{n=0}^\infty
\frac{(-x)^n}{n!}
\frac{\Gamma[(\gamma+n)/\eta]}{\Gamma(\gamma/\eta)}.
\label{eq:momentseries}
\end{equation}
For $\eta>1$ the series converges for all $x$ and defines a generalized Wright function in the classical Wright hierarchy \cite{Wright1940}; for $\eta=1$ it has radius $|x|<1$ and its analytic continuation reproduces the power kernel; for $0<\eta<1$ it is, in general, a divergent asymptotic expansion away from the origin. Representation \eqref{eq:AMB} remains the natural global form \cite{GorenfloLuchkoMainardi2007,BeghinCristofaroDaSilva2023,Mehrez2017,ParisKaminski2001}.

When $\eta=m\in\mathbb N$, the series can be regrouped into residue classes. Define
\begin{equation}
a_j=\frac{\gamma+j}{m},
\qquad
Z_m=\frac{(-1)^m x^m}{m^m},
\label{eq:ajZm}
\end{equation}
and
\begin{equation}
B_{m,j}=\left\{\frac{j+r}{m}:r=1,\ldots,m,\ r\ne m-j\right\},
\label{eq:Bmj}
\end{equation}
where the element equal to one is omitted. Then
\begin{equation}
A_{\gamma,m}(x)=
\sum_{j=0}^{m-1}
\frac{(-x)^j}{j!}
\frac{\Gamma(a_j)}{\Gamma(\gamma/m)}
{}_1F_{m-1}(a_j;B_{m,j};Z_m).
\label{eq:Ainteger}
\end{equation}
With $\lambda_j=(j+1)/m$,
\begin{equation}
\begin{aligned}
\Lmag_{\gamma,m,\beta}(\Fcal)
={}&\frac{1}{\beta}
\sum_{j=0}^{m-1}
\frac{(-1)^j x^{j+1}}{(j+1)j!}
\\
&\times
\frac{\Gamma(a_j)}{\Gamma(\gamma/m)}
{}_2F_m\!\left(
\begin{matrix}
a_j,\lambda_j\\
B_{m,j},\lambda_j+1
\end{matrix};
Z_m
\right).
\end{aligned}
\label{eq:Linteger}
\end{equation}
Coincident parameters in the upper and lower lists are to be canceled.

The first branches illustrate the hierarchy. For $\eta=1$,
\begin{equation}
\begin{aligned}
A_{\gamma,1}(x)
&=(1+x)^{-\gamma},
\\[2pt]
\Lmag_{\gamma,1,\beta}(\Fcal)
&=\frac{(1+x)^{1-\gamma}-1}
{\beta(1-\gamma)}.
\end{aligned}
\label{eq:eta1}
\end{equation}
This branch is known and contains the power kernels used in Ref.~\cite{FathiGuzmanVillanueva2026}.

For $\eta=2$, let
\begin{equation}
z=\frac{x^2}{4},
\qquad
a=\frac\gamma2,
\qquad
b=\frac{\gamma+1}{2},
\qquad
R_\gamma=\frac{\Gamma(b)}{\Gamma(a)}.
\end{equation}
The response can be written as
\begin{equation}
A_{\gamma,2}(x)=
2^{1-\gamma/2}\frac{\Gamma(\gamma)}{\Gamma(\gamma/2)}
\e^{x^2/8}D_{-\gamma}\!\left(\frac{x}{\sqrt2}\right),
\label{eq:Aeta2D}
\end{equation}
or, equivalently,
\begin{equation}
A_{\gamma,2}(x)=
{}_1F_1\!\left(a;\frac12;z\right)
-xR_\gamma\,{}_1F_1\!\left(b;\frac32;z\right).
\label{eq:Aeta2M}
\end{equation}
A closed antiderivative is
\begin{equation}
\begin{aligned}
\Lmag_{\gamma,2,\beta}
=\frac1\beta\bigg[&
 x\,{}_1F_1\!\left(a;\frac32;z\right) \\
&-\frac{R_\gamma x^2}{2}
 {}_2F_2\!\left(b,1;\frac32,2;z\right)
\bigg].
\end{aligned}
\label{eq:Leta2}
\end{equation}

For $\eta=3$, define
\begin{equation}
Z=-\frac{x^3}{27},
\quad
R_1=\frac{\Gamma[(\gamma+1)/3]}{\Gamma(\gamma/3)},
\quad
R_2=\frac{\Gamma[(\gamma+2)/3]}{\Gamma(\gamma/3)}.
\label{eq:eta3defs}
\end{equation}

Then
\begin{widetext}

\begin{equation}
A_{\gamma,3}(x)
=
{}_1F_2\!\left(
\frac{\gamma}{3};
\frac13,\frac23;
Z
\right)
-
xR_1\,
{}_1F_2\!\left(
\frac{\gamma+1}{3};
\frac23,\frac43;
Z
\right)
+
\frac{x^2R_2}{2}\,
{}_1F_2\!\left(
\frac{\gamma+2}{3};
\frac43,\frac53;
Z
\right).
\label{eq:Aeta3}
\end{equation}

\begin{equation}
\Lmag_{\gamma,3,\beta}
=
\frac{1}{\beta}
\left[
x\,{}_1F_2\!\left(
\frac{\gamma}{3};
\frac23,\frac43;
Z
\right)
-
\frac{x^2R_1}{2}\,
{}_1F_2\!\left(
\frac{\gamma+1}{3};
\frac43,\frac53;
Z
\right)
+
\frac{x^3R_2}{6}\,
{}_2F_3\!\left(
\frac{\gamma+2}{3},1;
\frac43,\frac53,2;
Z
\right)
\right].
\label{eq:Leta3}
\end{equation}

\end{widetext}
For $\eta\in\mathbb Q_+$ the Mellin--Barnes slopes are commensurate and the Gauss multiplication formula reduces Eq.~\eqref{eq:AMB} to a finite-order Meijer--$G$ function \cite{Meijer1946,Erdelyi1953,DLMF}. For arbitrary positive $\eta$, the Fox--$H$ function is the natural closed representation \cite{Fox1961,Braaksma1964,MathaiSaxenaHaubold2010}.

\section{Coupling to Einstein gravity and the general magnetic family}
\label{sec:einstein}

Once the matter sector has been reconstructed, we introduce gravity. On a purely magnetic background,
\begin{equation}
P=0,
\qquad
S=-u,
\qquad
u(r)=\frac{Q_m^2}{2r^4}\equiv\frac{C}{r^4}.
\label{eq:ur}
\end{equation}
To avoid typographical ambiguity, in this section we denote by $u$ the same positive argument previously denoted by $\Fcal$. Define
\begin{equation}
A_m(u):=\mathcal L_S(-u,0)>0,
\qquad
H_m(u):=-\mathcal L(-u,0).
\label{eq:AmHm}
\end{equation}
Since $H_m(0)=0$ and $H_m'(u)=A_m(u)$,
\begin{equation}
H_m(u)=\int_0^u A_m(v)\,\dd v.
\label{eq:HfromA}
\end{equation}

\begin{theorem}[Universal magnetic self-energy criterion]
For a nonvanishing magnetic charge, $Q_m\neq0$ and $C=Q_m^2/2>0$, the total electromagnetic energy
\begin{equation}
\Eem=4\pi\int_0^\infty r^2H_m(C/r^4)\,\dd r
\label{eq:Eemgeneral}
\end{equation}
is finite if and only if
\begin{equation}
{\displaystyle 
\int_0^\infty u^{-3/4}A_m(u)\,\dd u<\infty,
}
\label{eq:universalenergycriterion}
\end{equation}
and, when it converges,
\begin{equation}
\Eem=\frac{4\pi}{3}C^{3/4}
\int_0^\infty u^{-3/4}A_m(u)\,\dd u.
\label{eq:universalenergy}
\end{equation}
\end{theorem}
\begin{proof}
Insert Eq.~\eqref{eq:HfromA} into Eq.~\eqref{eq:Eemgeneral} and apply Tonelli's theorem. For fixed $u>0$, the inequality $u<C/r^4$ is equivalent to $r<(C/u)^{1/4}$. The radial integral of $r^2$ then gives $(C/u)^{3/4}/3$.
\end{proof}

If $A_m(u)\sim c u^{-\alpha}\ell(u)$ with $\ell$ slowly varying and, in addition, $uA_m'(u)/A_m(u)\to-\alpha$, finite self-energy generically requires $\alpha>1/4$, whereas magnetic causality $A_m+2uA_m'>0$ imposes $\alpha\le1/2$. The window $1/4<\alpha\le1/2$ therefore reappears as a universality class within this subclass; slowly varying factors can modify the lower boundary when they alter convergence of the exact self-energy criterion \cite{BinghamGoldieTeugels1987,RussoTownsendBH2026,FathiGuzmanVillanueva2026}.

\subsection{Einstein equations without a target geometry}

We begin with the most general static, spherically symmetric metric in areal-radius coordinates,
\begin{equation}
\dd s^2=-\e^{2\delta(r)}N(r)\dd t^2
+\frac{\dd r^2}{N(r)}+r^2\dd\Omega_2^2.
\label{eq:metricgeneral}
\end{equation}
For a magnetic monopole, $T^t{}_t=T^r{}_r$, and the Einstein equations imply $N\delta'=0$. On every connected static region, $\delta$ is constant and can be absorbed into a rescaling of $t$. We write
\begin{equation}
f(r)=N(r)=1-\frac{2Gm(r)}{r}.
\label{eq:fdef}
\end{equation}
The radial equation is
\begin{equation}
m'(r)=4\pi r^2H_m(C/r^4).
\label{eq:massEq}
\end{equation}
Therefore,
\begin{equation}
m(r)=M_0+4\pi\int_0^r x^2H_m(C/x^4)\,\dd x,
\label{eq:massQuadrature}
\end{equation}
where $M_0=m(0^+)$ when the self-energy is finite. The ADM mass is
\begin{equation}
M=M_0+\Eem.
\label{eq:MADM}
\end{equation}
Gravitational integration by quadratures is standard in Einstein--NLED \cite{PellicerTorrence1969,Bronnikov2001,Bronnikov2022}; what is specific here is that $H_m$ has been reconstructed beforehand, without using $f(r)$ as input data.

\subsection{Magnetic causality and horizon uniqueness}

Causality has a direct geometric consequence that does not depend on membership in $\cCX$. Define
\begin{equation}
\mathcal K_m(u):=2uA_m(u)-H_m(u).
\label{eq:Kmag}
\end{equation}
By Eq.~\eqref{eq:kappa},
\begin{equation}
\mathcal K_m'(u)=A_m(u)+2uA_m'(u)=A_m(u)\kappa(u).
\label{eq:Kprime}
\end{equation}

\begin{theorem}[Monotonicity under magnetic causality]
\label{thm:magneticmonotonicity}
Suppose $A_m>0$, $\kappa>0$ at every finite field strength, and the electromagnetic self-energy is finite. Define $M_0=M-\Eem$. If $M_0\ge0$, then
\begin{equation}
{\displaystyle f'(r)>0\qquad(r>0).}
\label{eq:fmonotone}
\end{equation}
Consequently, there is at most one positive horizon. If $M_0>0$, that horizon exists and is unique.
\end{theorem}
\begin{proof}
Since $\mathcal K_m(0)=0$ and Eq.~\eqref{eq:Kprime} is positive, $\mathcal K_m(u)>0$ for $u>0$. Define
\begin{equation}
D(r)=m(r)-rm'(r).
\end{equation}
Using Eq.~\eqref{eq:massEq} and $u'=-4u/r$,
\begin{equation}
D'(r)=8\pi r^2\mathcal K_m(u)>0.
\label{eq:Dprime}
\end{equation}
Finite self-energy and monotonicity of $H_m(C/r^4)$ imply $r^3H_m(C/r^4)\to0$ as $r\to0^+$: indeed,
\begin{equation}
\int_0^r s^2H_m(C/s^4)\,\dd s
\ge \frac{r^3}{3}H_m(C/r^4),
\end{equation}
and the left-hand side tends to zero. Thus $rm'(r)\to0$ and $D(0^+)=M_0\ge0$. For $r>0$, $D(r)>0$, and
\begin{equation}
f'(r)=\frac{2GD(r)}{r^2}>0.
\end{equation}
If $M_0>0$, then $f(r)\to-\infty$ as $r\to0^+$, while $f(r)\to1$ as $r\to\infty$. The intermediate value theorem together with monotonicity guarantees a unique positive root.
\end{proof}

Membership in $\cCX$ is not used in Theorem~\ref{thm:magneticmonotonicity}; only the stated magnetic causal and energetic hypotheses enter the proof.

\subsection{Central behavior of the generalized-gamma family}

The strong-field asymptotics of Eq.~\eqref{eq:AGG} is
\begin{equation}
A_{\gamma,\eta}(x)
\sim
\frac{\eta\Gamma(\gamma)}{\Gamma(\gamma/\eta)}x^{-\gamma},
\qquad x\to\infty.
\label{eq:AasymGG}
\end{equation}
Define
\begin{equation}
c_{\gamma,\eta}:=\frac{\eta\Gamma(\gamma)}{\Gamma(\gamma/\eta)}.
\end{equation}
Then
\begin{equation}
A_m(u)\sim c_{\gamma,\eta}\beta^{-\gamma}u^{-\gamma},
\end{equation}
\begin{equation}
H_m(u)\sim
\frac{c_{\gamma,\eta}}{1-\gamma}\beta^{-\gamma}u^{1-\gamma},
\end{equation}
and, for $\gamma>1/4$,
\begin{equation}
m(r)-M_0\sim
\frac{4\pi c_{\gamma,\eta}}{(1-\gamma)(4\gamma-1)}
\beta^{-\gamma}C^{1-\gamma}r^{4\gamma-1}.
\label{eq:masymGG}
\end{equation}
If $M_0=0$ and $1/4<\gamma<1/2$, the term $2Gm/r$ diverges positively and
\begin{equation}
f(0^+)=-\infty.
\end{equation}
Since $f'>0$, there is exactly one positive horizon even on the $M_0=0$ branch.

At the boundary $\gamma=1/2$,
\begin{equation}
c_{1/2,\eta}=\frac{\eta\sqrt\pi}{\Gamma\!\left(\frac{1}{2\eta}\right)},
\end{equation}
and
\begin{equation}
f(0^+)=1-8\sqrt2\,\pi G
\frac{\eta\sqrt\pi}{\Gamma\!\left(\frac{1}{2\eta}\right)}
\frac{|Q_m|}{\sqrt\beta}.
\label{eq:f0gammahalf}
\end{equation}
The critical charge on the $M_0=0$ branch is
\begin{equation}
{\displaystyle 
\left.\frac{|Q_m|}{\sqrt\beta}\right|_{\rm crit}
=\frac{\Gamma\!\left(\frac{1}{2\eta}\right)}{8\sqrt2\,\pi^{3/2}\eta G}.}
\label{eq:criticalGG}
\end{equation}
Above the critical value, $f(0^+)<0$ and there is a unique horizon; below it, the $M_0=0$ branch is horizonless. The critical point, where the limiting surface meets the central singularity, requires a separate global analysis, as in other classifications of causal NLED \cite{RussoTownsendBH2026,FathiGuzmanVillanueva2026}.

\section{Exact families of geometries and black holes}
\label{sec:exactBH}

The preceding formulas determine the geometry for any admissible response. To call a specific branch a \emph{black hole}, we additionally require a positive root of $f$. Theorem~\ref{thm:magneticmonotonicity} guarantees such a root for $M_0>0$ and for the additional branches identified in the previous subsection. Exactness of a metric by itself will not be conflated with automatic existence of a horizon throughout parameter space.

\subsection{The \texorpdfstring{$\cCX$}{CX} family by quadratures}

On the magnetic axis of Eq.~\eqref{eq:Lg},
\begin{equation}
z(r)=\frac{Q_m^2}{Tr^4},
\qquad
X=1+z,
\end{equation}
and
\begin{equation}
H_g(r)=\frac{T}{2}\int_0^{z(r)}g(1+t)\,\dd t.
\label{eq:Hg}
\end{equation}
The self-energy is finite if and only if
\begin{equation}
\int_0^\infty t^{-3/4}g(1+t)\,\dd t<\infty,
\label{eq:CXenergycriterion}
\end{equation}
and then
\begin{equation}
E^{(g)}_{\rm em}=\frac{2\pi}{3}T^{1/4}|Q_m|^{3/2}
\int_0^\infty t^{-3/4}g(1+t)\,\dd t.
\label{eq:CXenergy}
\end{equation}
Define
\begin{equation}
\mathcal A_g(z)=\int_0^z t^{-3/4}g(1+t)\,\dd t,
\qquad
\mathcal B_g(z)=\int_0^z g(1+t)\,\dd t.
\end{equation}
The electromagnetic energy exterior to radius $r$ is
\begin{equation}
E_g(r)=\frac{2\pi}{3}T^{1/4}|Q_m|^{3/2}\mathcal A_g(z)
-\frac{2\pi Tr^3}{3}\mathcal B_g(z),
\label{eq:Eg}
\end{equation}
and
\begin{equation}
f_g(r)=1-\frac{2G}{r}[M-E_g(r)].
\label{eq:fg}
\end{equation}
The weak-field limit reproduces Reissner--Nordstr\"om in our charge conventions,
\begin{equation}
f(r)=1-\frac{2GM}{r}+\frac{4\pi GQ_m^2}{r^2}+O(r^{-6}).
\label{eq:RNlimit}
\end{equation}

\subsection{Power laws and Born--Infeld: known controls}

For Eq.~\eqref{eq:Lpower},
\begin{equation}
H_q(r)=\frac{T}{2q}
\left[\left(1+\frac{Q_m^2}{Tr^4}\right)^q-1\right].
\label{eq:Hq}
\end{equation}
The known globally causal extension has $1/2\le q<1$; within our patch $D_T$, the endpoint $q=1$ discussed in Sec.~\ref{sec:CX} is also present. Finite self-energy selects, in either case,
\begin{equation}
\frac12\le q<\frac34.
\end{equation}
In this interval,
\begin{equation}
E^{(q)}_{\rm em}=\frac{2\pi}{3}T^{1/4}|Q_m|^{3/2}
\frac{\Gamma(1/4)\Gamma(3/4-q)}{\Gamma(1-q)}.
\label{eq:Epower}
\end{equation}
The exterior energy is
\begin{equation}
E_q(r)=\frac{2\pi Tr^3}{3q}
\left[1-{}_2F_1\!\left(-q,-\frac34;\frac14;-\frac{Q_m^2}{Tr^4}\right)\right],
\label{eq:Eq}
\end{equation}
and
\begin{equation}
f_q(r)=1-\frac{2G}{r}[M-E_q(r)].
\end{equation}
The choice $q=1/2$ reproduces the Einstein--Born--Infeld black-hole branch and its standard exact descendants \cite{GarciaSalazarPlebanski1984,Breton2003,FernandoKrug2003,Dey2004,CaiPangWang2004,RussoTownsendBH2026}; the limit $T\to\infty$ gives Reissner--Nordstr\"om. These branches serve as analytic controls and are not presented as new.

\subsection{Continuous Born--Infeld mixtures}

For a seed with scale $T$,
\begin{equation}
H_T(r)=T\left[\sqrt{1+\frac{Q_m^2}{Tr^4}}-1\right].
\end{equation}
Positivity allows the radial and spectral integrations to be interchanged. Therefore,
\begin{equation}
E^{\rm MBI}_\mu(r)=\int_0^\infty E_{\rm BI}(r;T)\,\mu(\dd T),
\end{equation}
\begin{equation}
f^{\rm MBI}_\mu(r)=1-\frac{2G}{r}[M-E^{\rm MBI}_\mu(r)].
\label{eq:fMBI}
\end{equation}
The self-energy is finite exactly when
\begin{equation}
\int_0^\infty T^{1/4}\mu(\dd T)<\infty.
\label{eq:MBImoment}
\end{equation}

\subsection{Exact solution for the \texorpdfstring{$\star$}{star} member}

For Eq.~\eqref{eq:gstar},
\begin{equation}
g_\star(1+t)=(1+t)^{-1/4}(1+t/2)^{-1/4}.
\end{equation}
The self-energy can be evaluated using the Euler representation of ${}_2F_1$ \cite{DLMF,GradshteynRyzhik2015}:
\begin{equation}
E^\star_{\rm em}=
\frac{2\sqrt\pi}{3}\Gamma\!\left(\frac14\right)^2
{}_2F_1\!\left(\frac14,\frac14;\frac12;\frac12\right)
T^{1/4}|Q_m|^{3/2}.
\label{eq:Estar}
\end{equation}
Defining $z=Q_m^2/(Tr^4)$,
\begin{equation}
\begin{aligned}
\mathcal A_\star(z)
&=4z^{1/4}
F_1\!\left(
\frac14;\frac14,\frac14;\frac54;
-z,-\frac z2
\right),
\\[2pt]
\mathcal B_\star(z)
&=z
F_1\!\left(
1;\frac14,\frac14;2;
-z,-\frac z2
\right).
\end{aligned}
\label{eq:ABstarAppell}
\end{equation}
and
\begin{equation}
E_\star(r)=\frac{2\pi}{3}T^{1/4}|Q_m|^{3/2}\mathcal A_\star(z)
-\frac{2\pi Tr^3}{3}\mathcal B_\star(z),
\label{eq:Estarr}
\end{equation}
\begin{equation}
{\displaystyle 
f_\star(r)=1-\frac{2G}{r}[M-E_\star(r)].}
\label{eq:fstar}
\end{equation}
For $M_0>0$, Theorem~\ref{thm:magneticmonotonicity} proves that this geometry has exactly one positive horizon.

\subsection{Arbitrary spectral measure}

The spectral representation can be coupled to Einstein gravity before a specific density is chosen. Assume throughout this subsection that Eq.~\eqref{eq:minusquarter} holds. In particular, $\rho(\{0\})=0$, so the factors $1/s$ below are harmless in the measure-theoretic sense. For $C=Q_m^2/2$, Tonelli's theorem gives
\begin{equation}
\begin{aligned}
m_\rho(r)=M_0
&+4\pi\int_0^\infty\frac{\rho(\dd s)}{s}\\
&\quad\times\left[
\frac{r^3}{3}
-\frac{(Cs)^{3/4}}{4}
\Gamma\!\left(-\frac34,\frac{Cs}{r^4}\right)
\right].
\end{aligned}
\label{eq:mrho}
\end{equation}
The metric is
\begin{equation}
f_\rho(r)=1-\frac{2Gm_\rho(r)}{r}.
\label{eq:frho}
\end{equation}
This formula is exact for every positive measure satisfying the assumptions of the construction. Horizon existence is then determined using the conditions of Theorem~\ref{thm:magneticmonotonicity} and, for $M_0=0$, the asymptotics of the particular measure.

\subsection{Generalized-gamma for all \texorpdfstring{$\eta>0$}{eta > 0}}

Let
\begin{equation}
x_r=\frac{\beta C}{r^4}.
\end{equation}
Integrating Eq.~\eqref{eq:LMB} over the radial exterior and choosing $0<c<\min(\gamma,1/4)$,
\begin{equation}
E_{\gamma,\eta,\beta}(r)=
\frac{4\pi C}{r\Gamma(\gamma/\eta)}
\frac{1}{2\pi i}\int
\frac{\Gamma(z)\Gamma[(\gamma-z)/\eta]}
{(1-z)(1-4z)}
 x_r^{-z}\,\dd z.
\label{eq:EGBM}
\end{equation}
In Fox--$H$ form,
\begin{equation}
\begin{aligned}
E_{\gamma,\eta,\beta}(r)
&=
\frac{\pi C}{r\Gamma(\gamma/\eta)}
H_{3,3}^{1,3}\!\left[
x_r\,\middle|\,
\begin{smallmatrix}
(1-\gamma/\eta,1/\eta),(0,1),(3/4,1)\\
(0,1),(-1,1),(-1/4,1)
\end{smallmatrix}
\right].
\end{aligned}
\label{eq:EGGFox}
\end{equation}
Therefore,
\begin{equation}
{\displaystyle 
f_{\gamma,\eta,\beta}(r)=
1-\frac{2G}{r}[M-E_{\gamma,\eta,\beta}(r)].}
\label{eq:fGG}
\end{equation}
Throughout the window \eqref{eq:GGwindow}, $M_0>0$ implies exactly one horizon. For $M_0=0$, $1/4<\gamma<1/2$ also implies a unique horizon; for $\gamma=1/2$ the threshold \eqref{eq:criticalGG} applies.

When $\eta=m\in\mathbb N$, define additionally
\begin{equation}
\mu_j=\frac{4j+1}{4m},
\qquad
Z_m(r)=\frac{(-1)^m x_r^m}{m^m}.
\end{equation}
The exterior energy reduces to the finite sum
\begin{widetext}
\begin{equation}
E_{\gamma,m,\beta}(r)=\frac{4\pi C}{r}
\sum_{j=0}^{m-1}
\frac{(-x_r)^j}{j!(j+1)(4j+1)}
\frac{\Gamma(a_j)}{\Gamma(\gamma/m)}
{}_3F_{m+1}\!\left(
\begin{matrix}a_j,\lambda_j,\mu_j\\
B_{m,j},\lambda_j+1,\mu_j+1
\end{matrix};Z_m(r)
\right).
\label{eq:Einteger}
\end{equation}
\end{widetext}
The branch $m=1$ reproduces the known power family. The branches $m\ge2$ provide a tower of distinct closed solutions; $m=2$ contains the parabolic-cylinder representation of the matter sector, while $m=3$ yields higher-order hypergeometric sums. For $\eta\in\mathbb Q_+$, Eq.~\eqref{eq:EGBM} reduces to a Meijer--$G$ function; for irrational $\eta$, Fox--$H$ naturally retains the incommensurate slopes.

\section{Optical cones and hyperbolicity of the \texorpdfstring{$\cCX$}{CX} class}
\label{sec:optical}

Propagation of discontinuities in Pleba\'nski NLED is governed by a Fresnel quartic that generically factorizes into two effective metrics \cite{Boillat1970,ObukhovRubilar2002,Schellstede2016}. Recent geometric formulations of the principal symbol reinforce the interpretation of the optical structure as the natural object organizing perturbation propagation in NLED \cite{GoulartBittencourt2026}. In gravitational applications with two invariants, polarization splitting can furthermore translate into distinct light rings and shadows \cite{dePaulaLimaCunhaHerdeiroCrispino2026}. For the class \eqref{eq:Lg}, it is convenient to define the standard invariants used in the Fresnel analysis,
\begin{equation}
\mathfrak F=\frac12F_{\mu\nu}F^{\mu\nu}=-2S,
\qquad
\mathfrak G=P,
\end{equation}
and
\begin{equation}
y=\frac{\mathfrak G}{T}.
\end{equation}
The required derivatives are
\begin{equation}
\begin{aligned}
\mathcal L_{\mathfrak F}
&=-\frac{g}{2},
&\qquad
\mathcal L_{\mathfrak G}
&=\frac{\mathfrak G}{T}g,
\\
\mathcal L_{\mathfrak F\mathfrak F}
&=-\frac{g'}{2T},
&
\mathcal L_{\mathfrak F\mathfrak G}
&=\frac{\mathfrak G}{T^2}g',
\\
\mathcal L_{\mathfrak G\mathfrak G}
&=\frac{g}{T}
-\frac{2\mathfrak G^2}{T^3}g'.
\end{aligned}
\label{eq:opticalderivs}
\end{equation}
Substitution into the general factorization gives reciprocal optical metrics that, up to positive conformal factors, can be written as
\begin{equation}
\mathfrak g_\pm^{\mu\nu}
=g^{\mu\nu}+\sigma_\pm F^\mu{}_{\lambda}F^{\nu\lambda},
\label{eq:opticalmetrics}
\end{equation}
with
\begin{equation}
\sigma_-=-\frac{1}{T(X+y^2)},
\label{eq:sigmaminus}
\end{equation}
\begin{equation}
\sigma_+=-\frac{2a(X)}{T[X+2a(X)y^2]}.
\label{eq:sigmaplus}
\end{equation}
Appendix~\ref{app:optical} gives the full algebraic substitution.

\begin{theorem}[Causal optical cones of $\cCX$]
For every member of $\cCX$ on the domain $X>0$:
\begin{enumerate}[label=(\roman*)]
\item $\sigma_-\le\sigma_+\le0$;
\item both optical metrics are nondegenerate and Lorentzian;
\item every covector timelike with respect to $g^{\mu\nu}$ is timelike with respect to both optical metrics.
\end{enumerate}
\end{theorem}
\begin{proof}
From $0\le a\le1/2$ and $X>0$, one immediately has $\sigma_+\le0$ and
\begin{equation}
\sigma_+-\sigma_-
=\frac{X(1-2a)}{T(X+y^2)(X+2ay^2)}\ge0.
\end{equation}
The nondegeneracy quantity for a metric of the form \eqref{eq:opticalmetrics} is
\begin{equation}
D=1+\sigma\mathfrak F-\sigma^2\mathfrak G^2.
\end{equation}
For the two roots,
\begin{equation}
D_-=
\frac{X}{(X+y^2)^2}>0,
\label{eq:Dminus}
\end{equation}
\begin{equation}
D_+=
\frac{X\,[X(1-2a)+2a+2a(1-2a)y^2]}
{(X+2ay^2)^2}>0.
\label{eq:Dplus}
\end{equation}
To make the signature statement explicit, consider
\[
\mathfrak g^{\mu\nu}(s)
=
g^{\mu\nu}
+
s\sigma_\pm F^\mu{}_{\lambda}F^{\nu\lambda},
\qquad
0\le s\le1.
\]
Its nondegeneracy factor is $D(s)=1+s\sigma_\pm\mathfrak F-s^2\sigma_\pm^2\mathfrak G^2$, which is concave in $s$. Since $D(0)=1$ and $D(1)=D_\pm>0$, one has $D(s)>0$ throughout the interval; by continuity, the signature cannot change from the Lorentzian signature at $s=0$. Finally, let $n_\mu$ be a covector timelike with respect to $g^{\mu\nu}$ and define $v^\lambda=F^{\mu\lambda}n_\mu$. By antisymmetry, $n_\lambda v^\lambda=0$, so $v^2\ge0$. Since $\sigma_\pm\le0$,
\begin{equation}
\mathfrak g_\pm^{\mu\nu}n_\mu n_\nu
=g^{\mu\nu}n_\mu n_\nu+\sigma_\pm v^2
<0.
\end{equation}
\end{proof}

Abalos \emph{et al.} showed that a Lagrangian NLED is symmetric hyperbolic if and only if the timelike cones of its two effective metrics have a nonempty intersection \cite{Abalos2015}. The theorem above exhibits such a common intersection explicitly.

\begin{corollary}[Hyperbolicity]
At the algebraic level, the electromagnetic subsystem of every theory in $\cCX$ is symmetric hyperbolic on $X>0$. If, in addition, $g\in C^2$ and the background and initial data have the regularity required by quasilinear symmetric-hyperbolic theory, the corresponding local Cauchy problem is well posed. In particular, the system is strongly hyperbolic.
\end{corollary}

The regularity clause promotes the pointwise cone criterion to the standard local well-posedness statement for quasilinear symmetric-hyperbolic systems \cite{Kato1975}. This is a statement about electromagnetic evolution, not about stability of the coupled Einstein--NLED background.

On a magnetic background,
\begin{equation}
X=1+\frac{Q_m^2}{Tr^4},
\end{equation}
and the angular factors of the two polarizations are
\begin{equation}
\kappa_- = \frac{1}{X},
\qquad
\kappa_+=1-\frac{2a(X)(X-1)}{X}.
\label{eq:kappapm}
\end{equation}
Born--Infeld, $a=1/2$, satisfies $\kappa_+=\kappa_-=1/X$, in agreement with its absence of birefringence \cite{RussoTownsendNoBirefringence2023,GibbonsHerdeiro2001}.

For the $\star$ member,
\begin{equation}
\kappa_-^\star=\frac1X,
\qquad
\kappa_+^\star=\frac{1+3X}{2X(1+X)},
\label{eq:kappastar}
\end{equation}
so that
\begin{equation}
0<\kappa_-^\star\le\kappa_+^\star\le1.
\end{equation}
The theory is birefringent, but both polarizations remain subluminal.

\section{Discussion and scope}
\label{sec:scope}

The inverse formulation separates two questions that are often conflated in NLED: whether a theory is admissible and whether it is uniquely selected. The deformation result shows that the causal Pleba\'nski region has genuine functional freedom even after Maxwell normalization and parity are imposed. Causality is therefore an exclusion principle rather than a uniqueness principle. The role of the affine-separable construction is correspondingly limited but useful: it selects a tractable sector in which the admissible functional freedom can be solved exactly, without promoting that sector to a classification of all causal NLED.

This distinction is especially relevant to self-duality. Modern self-dual constructions retain substantial functional freedom \cite{RussoTownsendCausalSelfDual2024,RussoTownsendDualities2024,RussoTownsendSimplified2025,KuzenkoRuhl2026}, whereas the standard no-birefringence problem places Born--Infeld in an exceptional position \cite{RussoTownsendNoBirefringence2023,MezincescuRussoTownsend2024}. The class $\cCX$ provides a complementary picture: Born--Infeld is the self-dual intersection of the affine-separable causal sector, but birefringent theories occupy its interior. Hence birefringence alone does not diagnose superluminal propagation or ill-posedness. The explicit optical cones make that statement geometric rather than perturbative.

The spectral construction resolves a different issue. Complete monotonicity is an additional assumption, not a consequence of two-invariant causality, but within that subclass it separates strong-field requirements into distinct properties of a positive measure. Self-energy depends on a negative moment, whereas magnetic causality constrains the way spectral weight approaches the origin. Their independence explains why changing the generalized-gamma shape can alter finite-field behavior without shifting the admissibility boundaries. More importantly, the lifting criterion shows that a magnetic law can satisfy every one-invariant causal test used here and nevertheless fail to extend to the affine-separable two-invariant sector. Axis data are therefore insufficient to infer full Pleba\'nski causality.

The gravitational theorem should be read as a sufficient mechanism for horizon uniqueness rather than as a classification of causal Einstein--NLED black holes. Recent work has established broader restrictions on causal NLED geometries and on Cauchy-horizon excision \cite{HaleHennigarKubiznak2026,RussoTownsendBH2026}. The result obtained here isolates a simple route to a stronger conclusion on the branch considered: positivity of the magnetic response and characteristic factor, together with finite self-energy and nonnegative residual mass, forces global monotonicity of the static metric function. Crucially, this statement no longer depends on the affine-separable ansatz. The exact families constructed in Sec.~\ref{sec:exactBH} are applications of that theorem, not the basis for it; this is why the gravitational result is more general than the explicit solutions used to illustrate it.

Several boundaries of the analysis are essential. The exact two-invariant characterization is restricted to the affine-separable domain $D_T$, and the spectral theorems require complete monotonicity on the magnetic axis. The gravitational construction assumes static spherical symmetry and a purely magnetic field, so it does not cover genuinely dyonic backgrounds with $P\neq0$. The optical-cone argument establishes symmetric hyperbolicity of the electromagnetic subsystem under the stated regularity assumptions; it does not establish linear or nonlinear stability of the coupled Einstein--NLED spacetime \cite{MorenoSarbach2003,NomuraYoshidaSoda2020,Chaverra2016,DeFeliceTsujikawa2025}. No rotating solution is inferred from the static family.

These limitations define concrete extensions rather than missing ingredients in the present proofs. A first problem is to decide whether magnetically admissible laws that fail the $\cCX$ lifting criterion admit causal completions in a different two-invariant sector or face a genuine obstruction. A second is to repeat the inverse selection with both invariants dynamically active in dyonic geometries. Coupled gravito-electromagnetic perturbations would test stability beyond symmetric hyperbolicity, while the closed optical metrics provide direct input for polarization-dependent light rings and critical curves \cite{dePaulaLimaCunhaHerdeiroCrispino2026}. Rotation is more restrictive still, because the constitutive relation, Einstein equations, and causal inequalities would have to remain compatible simultaneously rather than being imposed through a geometric deformation of the static metric.

\FloatBarrier

\section{Conclusions}
\label{sec:conclusions}

The main result is structural: imposing constitutive integrability, Maxwell normalization, and causal propagation before choosing a model does not determine a unique nonlinear electrodynamics. It defines a functional theory space. Within the affine-separable sector studied here, that freedom becomes analytically controllable and contains both the Born--Infeld point and genuinely non-self-dual, birefringent theories with well-behaved characteristic cones.

The magnetic and gravitational analyses show why this reordering is useful. A positive spectral representation converts finite self-energy and magnetic causality into independent measure-theoretic restrictions, while the lifting problem identifies the additional information required for a two-invariant completion. After Einstein coupling, the same matter-first logic yields a geometric statement that is independent of the explicit solvable family: under the stated causal and energetic hypotheses, the static metric function is monotone and cannot support more than one positive horizon. The optical analysis then verifies that the electromagnetic evolution problem remains symmetric hyperbolic throughout the causal domain of $\cCX$.

The resulting framework does not replace the general classification of causal Pleba\'nski theories or of Einstein--NLED black holes. Its contribution is to make a nontrivial sector of that problem constructive: admissible matter is selected first, exact families are generated within controlled assumptions, and consequences for self-energy, characteristic propagation, and horizon structure follow without prescribing the spacetime geometry. This establishes inverse reconstruction as a systematic route from local consistency conditions to global gravitational restrictions.

\begin{acknowledgments}
M.F. acknowledges financial support from the Agencia Nacional de Investigaci\'on y Desarrollo (ANID), Chile, through FONDECYT Postdoctoral Project No.~3260029. J.R.V. is partially supported by the Centro de F\'isica Te\'orica de Valpara\'iso (CeFiTeV).
\end{acknowledgments}

\section*{Data availability}

No new data were created or analysed in this study.

\appendix

\section{Characteristics of the \texorpdfstring{$(\Delta_d,p)$}{(Delta d,p)} reconstruction}
\label{app:characteristics}

Starting from Eq.~\eqref{eq:pPDE},
\begin{equation}
p_U-\frac{\Delta_d}{p^2}p_V=-\frac{\partial_V\Delta_d}{p},
\end{equation}
the characteristic equations are
\begin{equation}
\frac{\dd U}{\dd s}=1,
\qquad
\frac{\dd V}{\dd s}=-\frac{\Delta_d}{p^2},
\qquad
\frac{\dd p}{\dd s}=-\frac{\partial_V\Delta_d}{p}.
\end{equation}
Assuming $\Delta_d\in C^1$ and sufficiently regular positive data on a noncharacteristic curve, standard first-order quasilinear PDE theory guarantees a local classical solution as long as $p$ does not reach zero and characteristics do not cross \cite{CourantHilbert1962}. Defining $q=\Delta_d/p$, the original equation gives $p_U+q_V=0$, and the 1-form $p\,\dd V-q\,\dd U$ reconstructs $\mathcal L$. Causality of the reconstructed solution still requires Eqs.~\eqref{eq:causal1}--\eqref{eq:causal3}; the method of characteristics is a reconstruction tool, not a substitute for the physical filters.

\section{Causality details for the \texorpdfstring{$X$}{X} class}
\label{app:CXcausal}

Starting from Eq.~\eqref{eq:Xdef} and $\mathcal L_X=-Tg/2$, Eqs.~\eqref{eq:LSg}--\eqref{eq:LPPg} follow directly. For a symmetric $2\times2$ Hessian, positive semidefiniteness is equivalent to $\mathcal L_{SS}\ge0$, $\mathcal L_{PP}\ge0$, and $\det H\ge0$. If $g>0$ and $g'\le0$, then
\begin{equation}
\mathcal L_{SS}\ge0,
\qquad
\det H=-\frac{2gg'}{T^2}\ge0,
\end{equation}
and
\begin{equation}
\mathcal L_{PP}=\frac gT-\frac{2P^2}{T^3}g'\ge\frac gT>0.
\end{equation}
Necessity of $g'\le0$ follows from $\mathcal L_{SS}\ge0$. For the strong condition, Eq.~\eqref{eq:strongg} and the limit $\beta_T\to0^+$ at fixed $X$ give $g+2Xg'\ge0$. Finally,
\begin{equation}
a=-Xg'/g
\end{equation}
satisfies $a\ge0$ and
\begin{equation}
1-2a=\frac{g+2Xg'}{g}\ge0.
\end{equation}

\section{Details of the spectral bootstrap}
\label{app:spectralproofs}

Suppose that, for some $\gamma>0$,
\begin{equation}
\rho(s)\sim c\,s^{\gamma-1}\ell(s),
\qquad s\to0^+,
\end{equation}
with $\ell$ slowly varying. A Karamata Tauberian theorem gives \cite{BinghamGoldieTeugels1987}
\begin{equation}
A(\Fcal)\sim c\Gamma(\gamma)\Fcal^{-\gamma}\ell(1/\Fcal),
\qquad \Fcal\to\infty.
\end{equation}
Since $A$ is completely monotone, the monotone density theorem for regularly varying functions further implies
\begin{equation}
\frac{\Fcal A'(\Fcal)}{A(\Fcal)}
\longrightarrow -\gamma,
\qquad
\Fcal\to\infty.
\end{equation}
Therefore,
\begin{equation}
\kappa(\Fcal)
=
1+2\Fcal\frac{A'(\Fcal)}{A(\Fcal)}
\longrightarrow
1-2\gamma.
\end{equation}
If $\ell$ tends to a nonzero constant, magnetic causality requires $\gamma\le1/2$ and finite self-energy requires $\gamma>1/4$. At the boundary $\gamma=1/4$, slowly varying factors may alter convergence of the moment; hence the power-law rule does not replace the exact criterion \eqref{eq:minusquarter}.

For the generalized-gamma family, the change of variables $t=y/x$ in Eq.~\eqref{eq:AGG} gives
\begin{equation}
A_{\gamma,\eta}(x)
\sim
\frac{\eta\Gamma(\gamma)}{\Gamma(\gamma/\eta)}x^{-\gamma},
\end{equation}
independently confirming $\kappa\to1-2\gamma$. Identity \eqref{eq:virialGG} is stronger because it controls the sign of $\kappa$ for every finite $x$, not only asymptotically.

\section{Mellin--Barnes derivation and hypergeometric hierarchy}
\label{app:MB}

We begin with Eq.~\eqref{eq:AGG}. Following the standard Barnes/Mellin inversion framework \cite{Barnes1908,ParisKaminski2001}, for $0<c<\gamma$ use
\begin{equation}
\e^{-xt}=\frac{1}{2\pi i}\int_{c-i\infty}^{c+i\infty}\Gamma(z)(xt)^{-z}\,\dd z.
\end{equation}
The inner integral is
\begin{equation}
\int_0^\infty t^{\gamma-z-1}\e^{-t^\eta}\,\dd t
=\frac1\eta\Gamma\!\left(\frac{\gamma-z}{\eta}\right),
\end{equation}
which gives Eq.~\eqref{eq:AMB}. Integrating $x^{-z}$ from zero to $x$ introduces $x^{1-z}/(1-z)$ and yields Eq.~\eqref{eq:LMB}. A second radial integration uses
\begin{equation}
\int_r^\infty R^{4z-2}\,\dd R
=\frac{r^{4z-1}}{1-4z},
\qquad \Re z<\frac14,
\end{equation}
and leads to Eq.~\eqref{eq:EGBM}. The identities
\begin{equation}
\frac{1}{1-z}=\frac{\Gamma(1-z)}{\Gamma(2-z)},
\qquad
\frac{1}{1-4z}=\frac14\frac{\Gamma(1/4-z)}{\Gamma(5/4-z)}
\end{equation}
produce the Fox--$H$ forms \eqref{eq:LFoxH} and \eqref{eq:EGGFox}.

For $m\in\mathbb N$, split the series \eqref{eq:momentseries} into residue classes $n=mk+j$. The Gauss multiplication formula gives
\begin{equation}
(mk+j)!=j!\,m^{mk}k!
\prod_{\substack{r=1\\r\ne m-j}}^m
\left(\frac{j+r}{m}\right)_k,
\end{equation}
and
\begin{equation}
\Gamma\!\left(\frac{\gamma+j}{m}+k\right)=\Gamma(a_j)(a_j)_k.
\end{equation}
Substitution yields Eq.~\eqref{eq:Ainteger}. Integration with respect to $x$ introduces the pair $\lambda_j,\lambda_j+1$ and gives Eq.~\eqref{eq:Linteger}; the radial integration introduces $\mu_j,\mu_j+1$ and produces Eq.~\eqref{eq:Einteger}. For $\eta=p/q\in\mathbb Q_+$, applying the multiplication formula to both slopes in Eq.~\eqref{eq:AMB} converts the integrand into a finite product of unit-slope gamma functions, i.e. into a Meijer--$G$ function \cite{Meijer1946,Erdelyi1953,DLMF}.

\section{Born--Infeld mixtures and the power subfamily}
\label{app:BI}

For the mixture \eqref{eq:MBI}, positivity of the integrand allows the radial integration to be interchanged with the integration over the measure $\mu$ by Tonelli's theorem. Therefore,
\begin{equation}
E^{\rm MBI}_{\rm em}
=
\int_0^\infty
E^{\rm BI}_{\rm em}(T)\,\mu(\dd T),
\end{equation}
where the self-energy of a Born--Infeld seed with scale $T$ is
\begin{equation}
E^{\rm BI}_{\rm em}(T)=
\frac{2\sqrt{\pi}}{3}
\Gamma\!\left(\frac14\right)^2
|Q_m|^{3/2}T^{1/4}.
\end{equation}
Consequently,
\begin{equation}
E^{\rm MBI}_{\rm em}<\infty
\quad\Longleftrightarrow\quad
\int_0^\infty T^{1/4}\mu(\dd T)<\infty,
\end{equation}
which reproduces Eq.~\eqref{eq:MBImoment}.

For the power family, define
\begin{equation}
\zeta:=\frac{Q_m^2}{T}.
\end{equation}
The radial integral determining the exterior electromagnetic energy can be evaluated exactly as
\begin{equation}
\begin{aligned}
&\int_r^\infty x^2
\left[
\left(1+\frac{\zeta}{x^4}\right)^q-1
\right]\dd x
\\
&\qquad=
\frac{r^3}{3}
\left[
1-
{}_2F_1\!\left(
-q,-\frac34;\frac14;-\frac{\zeta}{r^4}
\right)
\right].
\end{aligned}
\end{equation}
Substituting $\zeta=Q_m^2/T$ gives Eq.~\eqref{eq:Eq} directly.

The total self-energy corresponds to the limit $r\to0^+$. Equivalently, a change of variables reducing the integral to beta-function form gives
\begin{equation}
E^{(q)}_{\rm em}
=
\frac{2\pi}{3}
T^{1/4}|Q_m|^{3/2}
\frac{
\Gamma(1/4)\Gamma(3/4-q)
}{
\Gamma(1-q)
},
\end{equation}
in agreement with Eq.~\eqref{eq:Epower}. Convergence at the central endpoint requires
\begin{equation}
q<\frac34.
\end{equation}

\section{Self-energy and Appell integrals for \texorpdfstring{$\mathcal L_\star$}{L star}}
\label{app:star}

With $g_\star(1+t)=(1+t)^{-1/4}(1+t/2)^{-1/4}$, the substitution $t=u/(1-u)$ gives
\begin{align}
&\int_0^\infty t^{-3/4}(1+t)^{-1/4}(1+t/2)^{-1/4}\,\dd t
\nonumber\\
&\qquad=
B\!\left(\frac14,\frac14\right)
{}_2F_1\!\left(\frac14,\frac14;\frac12;\frac12\right),
\end{align}
which yields Eq.~\eqref{eq:Estar}. For a finite upper limit,
\begin{align}
\int_0^z&t^{-3/4}(1+t)^{-1/4}(1+t/2)^{-1/4}\,\dd t
\nonumber\\
&=4z^{1/4}
F_1\!\left(\frac14;\frac14,\frac14;\frac54;-z,-\frac z2\right),
\end{align}
and
\begin{equation}
\int_0^z(1+t)^{-1/4}(1+t/2)^{-1/4}\,\dd t
=zF_1\!\left(1;\frac14,\frac14;2;-z,-\frac z2\right),
\end{equation}
with the standard Euler and Appell conventions \cite{DLMF,GradshteynRyzhik2015}.

\section{Derivation of the optical cones}
\label{app:optical}

With
\begin{equation}
X=1+\frac{\mathfrak F}{T}-\frac{\mathfrak G^2}{T^2},
\end{equation}
the derivatives are those in Eq.~\eqref{eq:opticalderivs}. The general Pleba\'nski quartic can be written as the product of two metrics of the form $g^{\mu\nu}+\sigma F^\mu{}_{\lambda}F^{\nu\lambda}$ \cite{ObukhovRubilar2002,Schellstede2016}. Substituting the derivatives and using $g'=-ag/X$, the quadratic polynomial in $\sigma$ is proportional to
\begin{equation}
[T(X+y^2)\sigma+1]
[T(X+2ay^2)\sigma+2a],
\qquad
y=\mathfrak G/T,
\end{equation}
which gives Eqs.~\eqref{eq:sigmaminus}--\eqref{eq:sigmaplus}. The nondegeneracy quantities follow by substituting each root into
\begin{equation}
D=1+\sigma\mathfrak F-\sigma^2\mathfrak G^2
\end{equation}
and using $\mathfrak F/T=X-1+y^2$, from which Eqs.~\eqref{eq:Dminus} and \eqref{eq:Dplus} follow.

\bibliography{refs2}

\end{document}